\documentclass[11pt, letterpaper]{amsart}
\usepackage{hyperref, color}

\newtheorem{thm}{Theorem}[section]

\newtheorem{lem}[thm]{Lemma}
\newtheorem{prop}[thm]{Proposition}
\theoremstyle{definition}
\newtheorem{defn}[thm]{Definition}

\newtheorem{ass}[thm]{Assumption}

\theoremstyle{remark}
\newtheorem{rem}[thm]{Remark}

\numberwithin{equation}{section}

\newcommand{\set}[1]{\left\{#1\right\}}
\newcommand{\Real}{\mathbb R}

\newcommand{\eps}{\varepsilon}

\newcommand{\such}{\ | \ }

\newcommand{\prob}{\mathbb{P}}

\newcommand{\filt}{\mathbb{F}}

\newcommand{\F}{\mathcal{F}}

\newcommand{\T}{\mathcal{T}}

\newcommand{\V}{\Theta}

\newcommand{\nada}[1]{}

\newcommand{\dfn}{\, := \,}

\newcommand{\pare}[1]{\left(#1\right)}
\newcommand{\bra}[1]{\left[#1\right]}
\newcommand{\cbra}[1]{\left\{#1\right\}}
\newcommand{\dbra}[1]{[\kern-0.15em[ #1 ]\kern-0.15em]}
\newcommand{\dbraco}[1]{[\kern-0.15em[ #1 [\kern-0.15em[}

\newcommand{\ito}{It\^o}

\newcommand{\tr}{\mathsf{Tr}}
\newcommand{\matdiv}[1]{\mathsf{div}\left(#1\right)}
\newcommand{\reals}{\ensuremath{\mathbb R}}
\newcommand{\loct}{\mathrm{Loc}}
\newcommand{\probgrowth}[2]{G\left(#1,#2\right)}

\newcommand{\ol}[1]{\overline{#1}}
\newcommand{\ul}[1]{\underline{#1}}

\newcommand{\D}{\mathcal{D}}
\newcommand{\sexp}[1]{\mathcal{E}\left(#1\right)}
\newcommand{\splex}{\Delta_{+}^{d-1}}
\newcommand{\osplex}{\Delta_{+, \leq}^{d-1}}
\newcommand{\tauinv}{\tau^{-1}}
\newcommand{\bigxord}{X^{()}}
\newcommand{\xord}{x^{()}}

\begin{document}

\title{Ergodic Robust maximization of asymptotic growth}

\author{Constantinos Kardaras}
\address{Constantinos Kardaras, London School of Economics, Statistics Department, 10 Houghton street, London, WC2A 2AE, UK}%
\email{k.kardaras@lse.ac.uk}%
\thanks{}

\author{Scott Robertson}
\address{Scott Robertson, Questrom School of Business, Boston University, 595 Commonwealth Ave.,  Boston, MA 02215, USA.}%
\email{scottrob@bu.edu}%
\thanks{C. Kardaras acknowledges support from the MC grant FP7-PEOPLE-2012-CIG, 334540; S. Robertson is supported in part by the National Science Foundation
  under grant number DMS-1613159.}

\begin{abstract}
We consider the problem of robustly maximizing the growth rate of investor wealth in the presence of model uncertainty. Possible models are all those under which the assets' region $E$ and instantaneous covariation $c$ are known, and where additionally the assets are stable in that their occupancy time measures converge to a law with density $p$.  This latter assumption is motivated by the observed stability of ranked relative market capitalizations for equity markets. We seek to identify the robust optimal growth rate, as well as a trading strategy which achieves this rate in all models. Under minimal assumptions upon $(E,c,p)$, we identify the robust growth rate with the Donsker-Varadhan rate function from occupancy time Large Deviations theory.  We also prove existence of, and explicitly identify, the optimal trading strategy.  We then apply our results in the case of drift uncertainty for ranked relative market capitalizations.  Assuming regularity under symmetrization for the covariance and limiting density of the ranked capitalizations, we explicitly identify the robust optimal trading strategy in this setting.
\end{abstract}

\date{\today}

\maketitle


\section*{Introduction}\label{S:I}

In this work, we identify portfolios which maximize the long-term growth rate of investor wealth in the presence of Knightian model uncertainty.  Optimal portfolios are robust, as they achieve the largest possible uniform growth across all models.  In our earlier work \cite{MR2985170}, beliefs ranged across models with common asset state space and instantaneous covariance; hence, model uncertainty was tantamount to lack of knowledge regarding the assets' drift. Presently, we obtain optimal portfolios when, in addition to the state space and covariance structure, assets are also ``stable'' in that their occupancy time measures converge to a known probability density.

Our motivation comes from stochastic portfolio theory, introduced in \cite{MR1861997}, \cite{MR1894767}, and in particular with the observation that the ranked relative market capitalizations (at least for equities in the United States) have remained remarkably stable over time.  As numerous subsequent articles have shown, such behaviour can be achieved by modelling market capitalizations via interacting diffusions, where interactions occur though the ranks. For example, \cite{MR2473654} considers Brownian particle systems with rank-dependent drifts, and proves ergodicity with limiting exponential distribution, for the process of spacings between ranked particles.  Extending the spacing analysis, \cite{MR3055264} proves stability of the ranked relative capitalizations, as well as long horizon estimates on the growth of certain classes of wealth processes.  For particular models, such as the Atlas model of \cite{MR1894767}, the authors of \cite{MR3055264} are able to explicitly identify the limiting density via its Laplace transform.

Given the observed stability of the ranked relative capitalizations, it is natural to ask how one may use this to achieve optimal growth for investor wealth.  Furthermore, as it is notoriously difficult to estimate asset drifts, can one essentially only use stability, covariance, and the stable capital distribution in order to derive optimal policies?

Broadly speaking, there are two approaches to answering these questions.  The first extends the notion of Cover's ``universal'' portfolio (see \cite{MR929084}), to construct portfolios which are growth optimal in a path-wise, model-free environment.  The second seeks to construct growth optimal portfolios that are \emph{functionally generated} in the sense of \cite{MR1894767}, producing optimal policies driven by functions of the underlying price process, and as such, are easily implemented using observations of the current state. Constructions of universal portfolios in relative capitalization models are treated in the recent articles \cite{wong2015universal}, \cite{cuchiero2016cover} (in fact, each of these treat functionally generated portfolios as well), while the functionally generated approach, aside from being pioneered in \cite{MR1861997}, \cite{MR1894767}, \cite{MR2676936}, \cite{MR2732837}, has been applied to long horizon problems in \cite{MR3055264}, \cite{MR2985170}, \cite{MR3114918}.

\smallskip

We follow the functionally generated approach, and provide here a brief sketch of the main argument. To treat both the ranked and unranked relative capitalization cases in a unified manner, we follow the abstract approach of \cite{MR2985170}, and assume the ``price'' process $X$ of a traded asset takes values in an arbitrary region $E\subseteq \reals^d$.  On the canonical space of $E$-valued continuous functions, we consider the class $\Pi$ of all probability measures $\prob$ under which:
\begin{itemize}
	\item $X$ is a semi-martingale with covariation $\int_0^\cdot c(X_t)dt$, where $c$ takes values in the set of positive definite symmetric $d \times d$ matrices.
	\item The laws of $\cbra{X_t; \, t\geq 0}$ are tight.
	\item $(1/T)\int_0^T h(X_u)du \rightarrow \int_E h(y)p(y)dy$ almost surely as $T \to \infty$, for all $h$ with $h^+\in L^1(E,p)$, where $p$ is a probability density on $E$.
\end{itemize}
   Wealth processes $V^{\vartheta}$ are constructed through $V^{\vartheta}=\mathcal{E}\left(\int_0^\cdot \vartheta_t'dX_t\right)$ for predictable strategies $\vartheta$ in the class $\V$ ensuring $X$-integrability under every $\prob \in \Pi$, with $\vartheta X$ denoting the proportion of wealth invested in $X$.  With the previous definitions, we seek to identify
\begin{equation}\label{eq: lambda_intro}
\lambda \dfn \sup_{\vartheta \in \V}\inf_{\prob\in\Pi} \probgrowth{V^{\vartheta}}{\prob},
\end{equation}
where $\probgrowth{V^{\vartheta}}{\prob}$ is the growth rate of $V^{\vartheta}$ in $\prob$-probability, precisely defined in \eqref{eq: growth} below. In addition to identifying $\lambda$, we seek a strategy $\hat{\vartheta} \in \V$ which achieves the growth rate $\lambda$ robustly across all $\prob\in\Pi$.  Our main result, Theorem \ref{thm: mr}, states that under minimal integrability assumptions, and one very important probabilistic assumption discussed below, we may conclude
\begin{equation}\label{eq: mr_intro}
\lambda = I,
\end{equation}
where $I = I(p)$ is the Donsker-Varadhan rate function, evaluated at $p$, associated to the second order linear operator $L^c =  (1/2)\tr(D^2 c)$ on $E$. Introduced in the series of papers \cite{MR0386024}, \cite{MR0428471}, \cite{MR690656} for ergodic Markov processes with generators $L$, the rate function $I$ governs large deviations for the occupancy time measures. Presently, we do not assume $L^c$ is ergodic (in fact, if $L^c$ were ergodic, $\lambda = 0$ as shown in Section \ref{S:Examples} below), but rather use the explicit form
\begin{equation}\label{eq: I_intro}
I = -\inf\cbra{ \int_E \frac{L^c u}{u}(y)p(y)dy \ \Big| \ u\in C^2(E), u>0, \frac{(L^c u)^+}{u} \in L^1(p)},
\end{equation}
as well as the interpretation of $I$ as the rate at which the occupancy time measures exit compact subsets of the space of $E$-valued probability measures---see \cite[Section 3]{MR2206349}. As expanded upon in Subsection \ref{SS:Heuristics} below, a short heuristic argument leads one to expect $\lambda=I$, provided their exists a function $\hat{u}$ such that $dX_t = (c\nabla \hat{u}/u)(X_t)dt + \sigma(X_t)dW_t$ is ergodic, with limiting density $p$ (here $\sigma$ is a square root of $c$).  However, as innocuous as this statement might seem, proving such a $\hat{u}$ exists for general (multi-dimensional) domains $E$, covariation functions $c$ and densities $p$ is a challenging task which takes up the bulk of the paper.  Interestingly, essentially the only $\hat{u}$ (up to a multiplicative constant) which can possibly lead to ergodicity is the optimizer of the right hand side of \eqref{eq: I_intro}.  Furthermore, $\hat{u}$ cannot lead to ergodicity without a-priori assuming the ``reversing'' diffusion $X^R$ with generator $L^R = (1/2)(\nabla\cdot (c\nabla) + (\nabla p/p)'c\nabla)$ is also ergodic. This follows from the remarkable results in \cite[Ch. 6]{MR1326606} on necessary and sufficient conditions for multi-dimensional diffusions to be transient or recurrent.

Provided $X^R$ is ergodic, under rather mild integrability assumptions (see Assumptions \ref{ass: Ecp} and \ref{ass: Ecp_i}) not only does \eqref{eq: mr_intro} hold, but also there exists an optimizer to the right hand side of \eqref{eq: I_intro} such that $dX_t = (c\nabla \hat{u}/\hat{u}) (X_t)dt + \sigma(X_t)dW_t$ is ergodic, and the functionally generated trading strategy $\hat{\vartheta}_\cdot = (\nabla \hat{u} / \hat{u}) (X_\cdot)$ is robust growth optimal, achieving growth rate $\lambda$ under all models in $\Pi$.  This is the statement of Theorem \ref{thm: mr}.

In Subsection \ref{SS:reversing_ass} we reinforce the importance of the ergodicity of $X^R$, by proving that without it, the robust growth optimal problem is in effect ill-posed. More precisely, if $X^R$ is not ergodic, then, at least in the one dimensional case, either $\Pi=\emptyset$ or $\lambda=\infty$.  Section \ref{S:Examples} contains numerous examples, proving, amongst other things, that both zero and infinite robust growth are possible for various choices of $(E,c,p)$. In particular, if the diffusion with generator $L^c$ is ergodic, strictly positive robust growth is impossible. This section also highlights the striking increase in complexity to the problem when going from one to multiple dimensions.

Section \ref{S:RBA} specifies the analysis to when the underlying process denotes relative market capitalizations. Here, there is a major subtlety to address: the observed phenomena is stability of the \emph{ranked} relative capitalizations, not the relative capitalizations themselves. However, trading does not happens in the ranked capitalizations, rather in the relative capitalizations and the market portfolio. Therefore, even though the natural inputs to the problem are the triple $(\osplex,\kappa,q)$; where $\osplex$ is the ordered unit simplex where the ranked relative capitalizations evolve, $\kappa$ a covariation function, and $q$ a density on $\osplex$, one must work on the unit simplex $\splex$ where the relative capitalizations live, and use a covariation function $c$ and density $p$ defined on this region.  To obtain $c$ and $p$ on $\splex$, we appropriately symmetrize $(\kappa,q)$. In order to apply the abstract theory, we ask in Assumption \ref{ass: kappa_q_extend} that such symmetrization preserves regularity in $(c,p)$.

Under Assumption \ref{ass: kappa_q_extend}, Proposition \ref{prop: growth_same} identifies the robust growth rate, as well as optimal strategy in the rank-based set up.  It also proves that optimal portfolios are functions solely of the ranked relative capitalizations, as one would expect.  The section then closes with a useful technical result which states that one can start with an arbitrary pair $(\kappa,q)$ on $\osplex$, and create a related pair which satisfies Assumption \ref{ass: kappa_q_extend} by only modifying $(\kappa,q)$ arbitrarily close near the boundary $\partial \osplex$.  Thus, our results allow for general covariances and densities on an arbitrary open subsets of $\osplex$. The price of the modification is that optimal policies are combinations of the equally weighted and market portfolios near where relative capitalizations cross ranks. However, an advantage of this modification is that it rules out sudden portfolio changes on capitalization crossings, which in practice would be infeasible over a long horizon, due to transactions costs.

The paper is organized as follows: Section \ref{S:APS} outlines the model, heuristic arguments and main result in the abstract setting.  Section \ref{S:Examples} contain examples, while Section \ref{S:RBA} specifies to the rank-based case. Appendix \ref{appsec:proof_of_main} contains the lengthy proof of the main abstract result, while Appendix \ref{S:proofs_from_rank} deals with proofs related to the rank-based model.

\section{Main Result}\label{S:APS}

\subsection{The problem}

At an abstract level, the problem we shall consider has three inputs: a region $E \subseteq \Real^d$ where an underlying stochastic process $X$ takes values; an instantaneous covariance function $c: E \mapsto \mathbb{S}^d_{++}$ for $X$, where $\mathbb{S}^d_{++}$ denotes the set of all symmetric strictly positive definite $d \times d$ matrices; and a ``limiting'' probability density $p$ for $X$. More precise discussion will appear after the following \emph{standing assumptions} on $(E,c,p)$; note that for the rest of the paper, all integrals over $E$ or its subsets are with respect to Lebesgue measure.
\begin{ass}\label{ass: Ecp}
For some fixed constant $\gamma \in (0,1]$, it holds that:
\begin{enumerate}
	\item $E = \bigcup_{n=1}^{\infty} E_n \subseteq\reals^d$, where for each $n$, $E_n$ is open, connected, bounded, and has $C^{2,\gamma}$ boundary. Furthermore, $\bar{E}_n \subset E_{n+1}$ and, if $d\geq 2$, $E_{n+1}\setminus\bar{E}_{n}$ is simply connected.
	\item $c\in C^{2,\gamma}(E; \mathbb{S}^d_{++})$.
	\item $p\in C^{2,\gamma}(E; (0,\infty))$ and $\int_E p = 1$.
\end{enumerate}
\end{ass}

As in \cite{MR2985170}, we work on the canonical space of continuous functions $\Omega = (C[0,\infty); E)$, equipped with its Borel sigma-algebra $\F$. The coordinate mapping process is denoted by $X$, and $\filt$ is the right-continuous enlargement of the natural filtration generated by $X$.

\begin{defn}\label{defn: Pi_p}
For a given density $p$ as in Assumption \ref{ass: Ecp}(3) above, $\Pi$ is the class of probability measures $\prob$ on $(\Omega,\F)$ such that:
\begin{enumerate}
\item $X_t\in E$, for all $t\geq 0$, $\prob$-a.s.
\item $X$ is a $\prob$-semimartingale with covariation process $[X, X] = \int_0^\cdot c(X_t) d t$, $\prob$-a.s.
\item For all Borel measurable functions $h$ on $E$ with $\int_E h^+ p < \infty$, it holds that
    \begin{equation*}
    \lim_{T\uparrow\infty} \frac{1}{T}\int_0^T h(X_t)dt = \int_E h p;\qquad \prob\textrm{-a.s.}
    \end{equation*}
\item The laws of $\{X_t; t \geq 0 \}$ under $\prob$ are tight.
\end{enumerate}
\end{defn}

Although condition (3) of Definition \ref{defn: Pi_p} above can be interpreted as $p$ being a limiting density for $X$ under $\prob \in \Pi $, we stress that we do not ask for any Markovian or stationary structure from the probabilities in $\Pi $. In fact, while condition (2) of Definition \ref{defn: Pi_p} implies that the instantaneous covariation is a function of the current state of $X$, the drift of $X$ under $\prob \in \Pi$ can be quite general, as long as the tenets of Definition \ref{defn: Pi_p} are satisfied.

We think of $X$ as an underlying process related to tradeable entities. In effect, $X$ denotes prices of securities, potentially denominated in units of another baseline asset, used for comparison. An example of such a denominating wealth is the total market capitalization, which would result in $X$ denoting relative capitalizations. (There is more discussion on this specific case of ``discounting'' in Section \ref{S:RBA}.)

In terms of trading in the previous financial environment, we shall use the following strategies.

\begin{defn}\label{defn: calV}
Under Assumption \ref{ass: Ecp}, and given the class $\Pi$ of Definition \ref{defn: Pi_p},  $\V$ is the class of predictable process that are $X$-integrable with respect to every $\prob \in \Pi$.
\end{defn}

For a process $\vartheta \in \V$ and measure $\prob \in \Pi$, we set
\begin{equation} \label{eq:wealth}
V^{\vartheta} \dfn \sexp{\int_0^\cdot \vartheta_t' d X_t}.
\end{equation}
Note that the version of $V^{\vartheta}$ may also depend on $\prob \in \Pi$, but we do not explicitly mention this dependency above, as it will be clear in each case which probability in $\Pi$ is considered. The interpretation is that $V^{\vartheta}$ is the wealth process generated starting from unit initial capital, and investing a proportion  $\vartheta^i_t X^i_t$ of current wealth in $X^i$ at time $t \geq 0$, for all $i = 1,\ldots,d$.

For $\vartheta\in\V$ and $\prob\in\Pi$, define
\begin{equation}\label{eq: growth}
\probgrowth{V^{\vartheta}}{\prob} \dfn \sup \cbra{ \gamma\in\reals \, : \, \lim_{T\uparrow\infty} \prob \bra{\frac{1}{T}\log V^{\vartheta}_T \geq \gamma} = 1}.
\end{equation}
As such, $\probgrowth{V^{\vartheta}}{\prob}$ is the \emph{long-run growth rate} (in probability) of the wealth generated by following the strategy $\vartheta$, when security prices evolve according to the probability measure $\prob$.  Our goal is to compute
\begin{equation}\label{eq: robust_problem}
\lambda\dfn \sup_{\vartheta\in\V}\inf_{\prob\in\Pi} \probgrowth{V^{\vartheta}}{\prob}.
\end{equation}
and obtain a robust maximizing strategy $\hat{\vartheta} \in \V$.

\subsection{Heuristics}\label{SS:Heuristics}

We first provide a heuristic argument for how the optimal strategy and robust growth rate are obtained.  To this end, and with ``$\tr$'' denoting the trace operator, set
\begin{equation}\label{eq: L_c}
L^c \dfn \frac{1}{2}\tr\pare{cD^2} \equiv \frac{1}{2} \sum_{i, j} c^{ij} \partial_{ij}.
\end{equation}
Note that $L^c$ is the second order operator associated to the driftless diffusion with covariance function $c$. Furthermore, while a solution to the \emph{generalized} martingale problem (c.f. \cite[Ch. 1]{MR1326606}) for $L^c$ on $E$ exists by Assumption \ref{ass: Ecp}, the solution may be exploding. Next, consider the class of functions
\begin{equation*}
\D \dfn \cbra{u \in C^2(E) \  \Big| \  u > 0, \, \int_E \left(\frac{L^c u}{u}\right)^+ p < \infty}.
\end{equation*}
With this notation, we define
\begin{equation*}
I \dfn - \inf_{u\in\D} \int_E \frac{L^c u}{u} p,
\end{equation*}
which is the Donsker-Varadhan rate function from occupancy time Large Deviations (LDP) theory evaluated at $p$. (In fact, $\D$ slightly differs from the domain used to prove occupancy time LDP in \cite[Chapter 4]{MR997938}, for example.)

Let $u\in \D$, and set $\vartheta^u_\cdot = (\nabla u/u)(X_{\cdot}) \in \V$. By \ito's formula, under $\prob\in\Pi$, it holds that
\begin{equation}\label{eq: v_u}
\frac{1}{T}\log(V^{\vartheta^u}_T) = \frac{1}{T} \log \left(\frac{u(X_T)}{u(X_0)}\right) - \frac{1}{T}\int_0^T \frac{L^c u}{u}(X_t)dt.
\end{equation}
Thus, under Assumption \ref{ass: Ecp}, it follows that $G(V^{\vartheta^u}, \prob) = -\int_E (L^c u/u)p$. As this holds for all $u\in\D$ and $\prob\in\Pi$, by \eqref{eq: robust_problem} we obtain
\begin{equation} \label{eq: lambda_lb}
\lambda \geq \sup_{u\in \D} \inf_{\prob\in\Pi} G(V^{\vartheta^u}, \prob) = \sup_{u\in \D} \cbra {-\int_E \frac{L^c u}{u}p} = I.
\end{equation}

Now, let $\sigma$ denote the unique positive definite symmetric square root of $c$, and assume that, for some $\hat{u}\in \D$ and an appropriate initial condition $X_0 \in E$, the diffusion with dynamics
\begin{equation}\label{eq: p_star_diffusion}
d\hat{X}_t = c\frac{\nabla \hat{u}}{\hat{u}}(\hat{X}_t)dt + \sigma(\hat{X}_t)d\hat{W}_t,
\end{equation}
is ergodic with invariant density $p$.  This implies the probability measure $\hat{\prob}$ induced by the law of $\hat{X}$ is in $\Pi$.  The wealth process $\hat{V}$ obtained by $\hat{\vartheta}_\cdot = (\nabla \hat{u} / \hat{u})(X_{\cdot})$ is in $\V$ and is growth-optimal for the model $\hat{\prob}$. Therefore, \eqref{eq: robust_problem} gives
\begin{equation*}
\lambda \leq \sup_{\vartheta\in \V} G(V^\vartheta,\hat{\prob}) = G(V^{\hat{\vartheta}},\hat{\prob}) = -\int_E \frac{L^c \hat{u}}{\hat{u}}p \leq \sup_{u\in \D} \cbra {-\int_E \frac{L^c u}{u}p} = I.
\end{equation*}
Note that, if the discussion of this paragraph is valid, then\emph{ a posteriori} $\hat{u}$ has to be a minimizer of $\D \ni u \mapsto \int_E \pare{L^c u / u } p$. We also regard $\hat{\prob}$ as a ``worst-case'' model, in the sense that the maximal growth achievable under $\hat{\prob} \in \Pi$ is $\lambda $.

\smallskip

From the above discussion, we are led to conjecture that $\lambda = I$.  As $I \leq \lambda$ follows from \eqref{eq: lambda_lb}, the difficulty is in establishing existence of a minimizer $\hat{u} \in \D$ of the mapping $\D \ni u \mapsto -\int_E \pare{L^c u / u } p$, and showing that the corresponding diffusion in \eqref{eq: p_star_diffusion} is ergodic with invariant measure $p$.

\subsection{The main result}

In order to carry out the plan outlined in Subsection \ref{SS:Heuristics}, we must make additional assumptions on how $(E,c,p)$ interact. To simplify the presentation, set
\begin{equation}\label{eq: ell_def}
\ell \dfn \frac{\nabla p}{p} + c^{-1}\matdiv{c}, \quad \text{where} \quad \matdiv{c}^i = \sum_j \partial_j c^{ij}, \quad i=1,\dots,d.
\end{equation}

\begin{ass}\label{ass: Ecp_i}
The following hold:
\begin{enumerate}
\item[(i)] $\int_E \ell'c\ell p < \infty$.
\item[(ii)] $\int_E \left(\nabla\cdot (pc\ell))\right)^+ < \infty$.
\item[(iii)] For the symmetric second order linear operator
\begin{equation}\label{eq: L_R_def}
L^R \dfn \frac{1}{2}\nabla\cdot(c\nabla) + \frac{1}{2}\frac{\nabla p}{p}'c\nabla\  =  \frac{1}{2}\tr\pare{cD^2} + \frac{1}{2}c\ell'\nabla,
\end{equation}
a (non-explosive) solution to the Martingale problem for $L^R$ on $E$ exists.
\end{enumerate}
\end{ass}

\begin{rem}\label{rem: p_c}
Recall that $\sigma$ denotes the unique positive definite symmetric square root of $c$. The diffusion $X^R$ associated to $L^R$ has dynamics
\begin{equation}\label{eq: reversing_diffusion}
dX^R_t = \frac{1}{2}\left(c\frac{\nabla p}{p} + \matdiv{c}\right)(X^R_t)dt + \sigma(X^R_t)dW_t.
\end{equation}
For any Brownian motion $W$ (on some probability space), a combination of Assumption \ref{ass: Ecp} and Assumption \ref{ass: Ecp_i}(iii) implies there exists a unique strong solution for any initial condition $X^R_0 \in E$. Furthermore, as formally $p$ is a candidate invariant density for $X^R$, Assumption \ref{ass: Ecp_i}(iii) also implies the seemingly stronger result that $X^R$ is ergodic with invariant density $p$: see \cite[Corollary 4.9.4]{MR1326606}.

Given the discussion in Subsection \ref{SS:Heuristics}, one might believe \eqref{eq: reversing_diffusion} yields the ``worst case model'' $\hat{\prob}$ leading to dynamics \eqref{eq: p_star_diffusion}, and this is why we need it to be non-exploding. This is \emph{not} the case in general---Remark \ref{rem:phi} and \S \ref{SS:gradient} will make it clear that \eqref{eq: reversing_diffusion} is a worst case model if and only if $c^{-1}\matdiv{c}$ is a gradient. Nevertheless, there is a very good reason why we enforce Assumption \ref{ass: Ecp_i}(iii). As shown in \cite[Theorem 6.6.2 (ii)]{MR1326606}, if Assumption \ref{ass: Ecp_i}(iii) fails, then there are \emph{no} time-homogeneous diffusions whose laws are in $\Pi$. Therefore, \emph{a fortiori}, the candidate for the ``worst case model'' of \eqref{eq: p_star_diffusion} will not belong to $\Pi$, making it impossible to prove Theorem \ref{thm: mr} that follows. In fact, if Assumption \ref{ass: Ecp_i}(iii) fails to hold, it is not clear whether the class $\Pi$ contains any elements whatsoever, and even if it did, it is also not clear if the robust problem is well-posed. To reinforce these points, Proposition \ref{prop: one_d_reversing_explode} below will have more to reveal for the one-dimensional case.
\end{rem}

\begin{rem} \label{rem: non_equiv_asss}
Neither of conditions (i) and (iii) in Assumption \ref{ass: Ecp_i} implies the other.  That (i) does not imply (iii) follows by letting $E = (0,1)$, and $p(x) = c(x) = 1$. To show that (iii) does not imply (i), let $E=(0,\infty)$, and $p(x) = B e^{-B x}, c(x) = \xi^2 x$ for $B,\xi > 0$. Then,
\begin{equation*}
\int_E \ell'c\ell p = \int_0^\infty B \xi^2x e^{-B x}\left(\frac{1}{x}-B\right)^2 = \infty,
\end{equation*}
but $X^R$ has dynamics $dX^R_t = (1/2)\left(\xi^2 - B X^R_t\right)dt + \xi\sqrt{X^R_t} d W_t$ and hence is non-explosive from the well-known properties of the CIR process.
\end{rem}

What follows is our main result, the proof of which is the purpose of Appendix \ref{appsec:proof_of_main}.

\begin{thm}\label{thm: mr}
Let Assumptions \ref{ass: Ecp} and \ref{ass: Ecp_i} hold. Then, there exists a unique (up to a multiplicative constant) $\hat{u} \in \D$ such that
\begin{equation} \label{eq:minimization}
\hat{u} = \arg \min_{u \in \D} \int_E \frac{L^c u}{u} p.
\end{equation}
Furthermore, it holds that
\begin{equation}\label{eq: lambda_I_equiv}
\lambda = I = \frac{1}{2} \int_E \pare{\frac{\nabla \hat{u}}{\hat{u}}}' c \pare{\frac{\nabla \hat{u}}{\hat{u}}}p,
\end{equation}
and the trading strategy
\begin{equation} \label{eq:trading_strategy}
\hat{\vartheta}_\cdot = \frac{\nabla \hat{u}}{\hat{u}} (X_\cdot)
\end{equation}
is such that $\probgrowth{V^{\hat{\vartheta}}}{\prob} = \lambda$, for all $\prob\in\Pi$.
\end{thm}

\begin{rem}\label{rem:phi}
The proof of Theorem \ref{thm: mr} will show that $\hat{u} = \exp(\hat{\phi} / 2)$, where $\hat{\phi} : E \mapsto \Real$ is the unique (up to an additive constant) function such that
\[
\hat{\phi} = \arg \min_{\phi \in C^2(E)} \int_{E} \left(\nabla\phi -\ell\right)'c\left(\nabla\phi -\ell \right)p.
\]
This variational problem is sometimes a more convenient problem to solve then identifying $\hat{u}$ via the right hand side of \eqref{eq:minimization}.
\end{rem}

\subsection{On Assumption \ref{ass: Ecp_i}(iii)}\label{SS:reversing_ass}

We elaborate here on the importance of Assumption \ref{ass: Ecp_i}(iii), already hinted in Remark \ref{rem: p_c}, by investigating deeper the one-dimensional case.

\begin{prop}\label{prop: one_d_reversing_explode}
	
	Assume that $d=1$ and $E=(\alpha,\beta)$ for $-\infty\leq \alpha < \beta \leq \infty$. Let $(c,p)$ satisfy Assumption \ref{ass: Ecp} and Assumptions \ref{ass: Ecp_i} (i) and (ii). Then, if Assumption \ref{ass: Ecp_i} (iii) fails, it either holds that $\Pi = \emptyset$ or $\lambda = \infty$.
	
\end{prop}

\begin{proof}
Assume that $\Pi\neq\emptyset$. Let $x_0 \in (\alpha,\beta)$. From \cite[Theorem 5.1.1]{MR1326606}, Assumption \ref{ass: Ecp_i}(iii) failing is equivalent to either $\int_\alpha^{x_0} 1/(pc) < \infty$ or $\int_{x_0}^\beta 1/(pc) < \infty$.  We shall only consider the case where $\int_\alpha^\beta 1/(pc) <\infty$ as the other cases are similar.  To this end, from \eqref{eq: v_u} with $u = \sqrt{pc}$, it follows that
\begin{equation*}
\frac{L^c u}{u}p = \frac{1}{2}\ddot{(pc)} - \frac{1}{4}\frac{\dot{(pc)}^2}{4pc},
\end{equation*}
and hence Assumption \ref{ass: Ecp_i}(ii) implies $(L^c u/u)^+ \in L^1(E,p)$ so that $u \in \D$. Next, from \eqref{eq: lambda_lb}, it is clear that if $(L^c u/u)^- \not\in L^1(E,p)$ then $\lambda = \infty$.  If $(L^c u/u)^- \in L^1(E, p)$, which, along with Assumption \ref{ass: Ecp_i}(ii), implies that $(L^c u/u) \in L^1(E,p)$, for $\eps > 0$ consider the function
\begin{equation*}
v_\epsilon (x) \dfn \sqrt{\eps + \int_\alpha^x \frac{1}{pc}}, \quad x \in (\alpha, \beta).
\end{equation*}
A straightforward shows that
\begin{equation*}
\frac{L^c(u v_\epsilon)}{u v_\epsilon} = \frac{L^c u}{u} + \frac{L^Rv_\epsilon}{v_\epsilon} = \frac{L^c u}{u}  - \frac{1}{8}\frac{c}{(pc)^2\left(\eps + \int_\alpha^x \pare{pc}^{-1}\right)^2}.
\end{equation*}
It thus follows that $(L^c(u v_\epsilon)/(u v_\epsilon))^+ \in L^1(E,p)$ so that $u v_\epsilon \in \D$. Furthermore,
\begin{align*}
-\int_\alpha^\beta \frac{L^c(uv_\epsilon)}{uv_\epsilon}p &= -\int_\alpha^\beta \frac{L^c u}{u}p + \frac{1}{8}\int_\alpha^\beta \frac{1}{pc\left(\eps + \int_\alpha^x \pare{pc}^{-1}\right)^2};\\
&=-\int_\alpha^\beta \frac{L^c u}{u}p + \frac{1}{8}\left(-\frac{1}{\eps + \int_\alpha^x \pare{pc}^{-1}} \bigg|_{x=\alpha}^{x=\beta}\right);\\
&=-\int_\alpha^\beta \frac{L^c u}{u}p - \frac{1}{8}\frac{1}{\eps + \int_\alpha^\beta \pare{pc}^{-1}} + \frac{1}{8\eps}.
\end{align*}
Thus, we see from \eqref{eq: lambda_lb} and $(L^c u/u) \in L^1(E,p)$ that 
\begin{equation*}
\lambda \geq \lim_{\eps\downarrow 0} \left(-\int_\alpha^\beta \frac{L^c u}{u}p - \frac{1}{8}\frac{1}{\eps + \int_\alpha^\beta \pare{pc}^{-1}} + \frac{1}{8\eps}\right) = \infty,
\end{equation*}
concluding the proof.
\end{proof}

\section{Examples}\label{S:Examples}

\subsection{One-dimensional case}
\label{SS:one-dim_case}
In the one dimensional case, $E = (\alpha,\beta)$ for $-\infty\leq\alpha < \beta \leq \infty$, and in view of Remark \ref{rem:phi}, we have $\ell = \dot{(pc)}/pc$ and $\hat{u}  = \sqrt{pc}$. Here, using \cite[Theorem 5.1.5]{MR1326606}, it easily follows that Assumptions \ref{ass: Ecp} and \ref{ass: Ecp_i} hold provided that:
\begin{itemize}
	\item $\int_\alpha^\beta (\dot{pc})^2/pc < \infty$;
	\item for some $x_0 \in E$ we have $\lim_{x\downarrow\alpha} \int_x^{x_0} (pc)^{-1} = \infty = \lim_{x\uparrow\beta} \int_{x_0}^x (pc)^{-1}$;
	\item $\int_\alpha^\beta (\ddot{pc})^+ < \infty$.
\end{itemize}

\subsection{The ``gradient'' case}
\label{SS:gradient}

Assume $c$ satisfies the special condition
\begin{equation}\label{eq: gradient_condition}
c^{-1}\matdiv{c} = \nabla H,
\end{equation}
for some function $H\in C^{1,\gamma}(E;\reals)$.  Note this always holds in the one-dimensional case as $H = \log(c)$.  Here, Remark \ref{rem:phi} gives $\hat{\phi} = \log(p) + H$, i.e., $\hat{u} = \sqrt{p} \exp(H / 2)$.  Also, ergodicity under the candidate worst-case model holds directly by Assumption \ref{ass: Ecp_i}(iii), since in this case the reversing diffusion of Assumption \ref{ass: Ecp_i}(iii) is in fact the worst-case model.

\subsection{Robust growth and explosion} Recall the operator $L^c$ from \eqref{eq: L_c}.  If the diffusion associated to $L^c$ explodes in finite time, then robust growth is achievable for all densities $p$ such that Assumptions \ref{ass: Ecp}, \ref{ass: Ecp_i} hold.  Indeed, it follows from \cite[Lemma 3.5]{MR2206349} and \cite[Lemma 33]{MR2932547} that for all densities $p$ one has
\begin{equation*}
0 < -\inf_{u \in \tilde{\D}}\int_E \frac{L^c u}{u}p \leq -\inf_{u\in \D}\int_E \frac{L^c u}{u}p = I,
\end{equation*}
where $\tilde{D}\subset \D$ contains those $u\in C^2(E), u>0$ such that $L^c u/u$ is bounded from above. Thus, provided the requisite integrability assumptions on $p$ we know that $\lambda = I > 0$.

\subsection{Cox-Ingersoll-Ross model under uncertainty} Let $E=(0,\infty)$ and $c(x) = \xi^2x$ for $\xi >0$. For $A>1$ and $B>0$ set
\begin{equation}\label{eq: CIR_p}
p(x) = \frac{B^A}{\Gamma(A)}x^{A-1}e^{-Bx};\qquad \Gamma(A) = \int_0^\infty y^{A-1}e^{-y}dy.
\end{equation}
Assumptions \ref{ass: Ecp} and \ref{ass: Ecp_i} hold, and a straightforward calculation using $\hat{u}=\sqrt{pc}$ (c.f. Subsection \ref{SS:one-dim_case}) shows
\begin{equation*}
\lambda =  \lambda(A,B) = \frac{\xi^2 AB}{8(A-1)}.
\end{equation*}
In \cite{MR2985170}, we considered a wider class of models, where the region $E$ and covariation function $c$ are known, but no assumption is made regarding a limiting density $p$ (not even whether such a density exists). In this example, since $\lim_{B\downarrow 0} \lambda(A,B) = 0$, there is no possibility to achieve strictly positive growth in setting of \cite{MR2985170}. However, once a measure is specified, strictly positive growth is possible.


\subsection{Infinite robust growth}

We continue with the model of the previous subsection, with the twist that here we assume that $A=1$ so that $p$ from \eqref{eq: CIR_p} is $p(x) = Be^{-Bx}$.  As shown in Remark \ref{rem: non_equiv_asss}, Assumption \ref{ass: Ecp_i}(i) does not hold. However, Assumption \ref{ass: Ecp_i}(iii) does hold, and this implies $\Pi\neq\emptyset$ (c.f. \cite{MR1152459} for verification of item $(3)$ with unbounded functions).  In this model, it is possible to obtain infinite robust growth, as we now show. Motivated by Subsection \ref{SS:one-dim_case}, we set $u(x) = \sqrt{pc}(x) = \xi \sqrt{B x} \exp(- B x / 2)$.  Calculation shows $L^c u/u = - (1/8)\xi^2(1/x + 2B - B^2x)$ so that $(L^c u/u)^+\in L^1(E,p)$ and hence $u\in\D$.  But, it is clear $(L^c u/u)^- \not\in L^1(E,p)$ and hence \eqref{eq: v_u} shows for all $\prob\in\Pi$ that for $\vartheta^u = (\nabla u / u) (X_\cdot)$,  $\probgrowth{V^{\vartheta^u}}{\prob} = \infty$ and hence $\lambda = \infty$.

\subsection{Zero robust growth}  Let $E=\reals$. Let $c$ be any positive smooth function such that $\int_{\reals} (1/c) = 1$ and set $p=1/c$.  From \cite[Ch. 5]{MR1326606} we deduce the diffusion with dynamics $dX_t = \sqrt{c(X_t)}dW_t$ is positive recurrent with invariant measure $p$, and hence we can take $\hat{u}\equiv 1$ in \eqref{eq: p_star_diffusion}.  As such, the trading strategy $\hat{\vartheta} = \nabla \hat{u}/\hat{u} \equiv 0$ achieves maximal growth under $\hat{\prob} \in \Pi $. But, not trading trivially leads to $\probgrowth{V^{\hat{\vartheta}}}{\hat{\prob}} = 0$; hence, $\lambda = 0$.

\section{An Application in Ranked-Based Models}\label{S:RBA}

\subsection{Relative market capitalizations}

To motivate the results of this section, start with a collection $S = (S^i; \, i=1, \ldots, d)$ of processes representing market capitalizations of $d$ stocks, and set $M = \sum_{i=1}^d S^i$ as the total capitalization of this particular (sub)market. Then, with $X^i \dfn S^i / M$, $i =1, \ldots, d$, the process $X = (X^i; \, i =1, \ldots, d)$ denotes \emph{relative} market capitalizations. We assume that no stock capitalization vanishes, so that $X$ takes values in the open simplex
\begin{equation}\label{E:simplex}
\splex \dfn \cbra{x \in \reals^d \ \Big| \ \min_{i=1,\ldots,d} x^i >0, \ \sum_{i=1}^d x^i = 1}.
\end{equation}

As already noted in Section \ref{S:APS}, wealth from investment (as well growth rates) will \emph{not} be absolute, but rather relative to market capitalization. In fact, investment is defined with respect to the relative capitalizations $X$, and not with respect to the original prices $S$, through the usual \emph{change-of-num\'eraire} technique. As direct calculations show, for any $d$-dimensional predictable process $\pi$ with $\sum_{i=1}^d \pi^i =1$,
\begin{equation*}
\frac{d U_t}{U_t} = \sum_{i=1}^d \pi^i_t \frac{d S^i_t}{S^i_t} \quad \Longleftrightarrow \quad \frac{d(U/M)_t}{(U/M)_t} = \sum_{i=1}^d \pi^i \frac{dX^i_t}{X^i_t}.
\end{equation*}
To wit, if a strategy of portfolio weights $\pi = (\pi^i; \, i=1, \ldots, d)$ is fully invested in stocks, the same strategy applied to the relative capitalizations results in wealth relative to the total market capitalization. Note also that, as the vector-valued process $X$ is degenerate (the sum of its components equals one), there is no loss of generality in assuming the strategies $\vartheta$ resulting in \eqref{eq:wealth} are such that $\sum_{i=1}^d X^i\vartheta^i = 1$. Indeed, for any predictable strategy $\vartheta$, if one defines a strategy $\eta$ via $\eta^i = \vartheta^i + (1 - \sum_{j=1}^d X^j\vartheta^j)$ for $i =1, \ldots, d$, then we have $\sum_{i=1}^d X^i\eta^i = 1$, and
\[
\sum_{i=1}^d \eta_t^i dX^i_t = \sum_{i=1}^d \vartheta_t^i dX^i_t + \pare{1 - \sum_{j=1}^d X^j\vartheta^j_t} \sum_{i=1}^d dX^i_t = \sum_{i=1}^d \vartheta_t^i dX^i_t,
\]
implying that $V^\vartheta = V^\eta$.

\subsection{Ranked capitalizations}

As has been observed in \cite[Section 5]{MR1894767}, empirical time-series data suggest that the capital distribution curve (i.e. the log-log plot of ranked relative capitalizations versus rank, in decreasing order) is stable for U.S. equities. This leads to the introduction of so-called \emph{ranked based} models for financial markets. Here, we shall not go into the details of ranked based models; for a thorough treatment, see \cite[Sections 4,5]{MR1894767}. Rather, we introduce assumptions on the ranked capitalizations, as opposed to the actual capitalizations, and consider questions of robust growth.

Define the ordered simplex
\begin{equation}\label{E:ordered_simplex}
\osplex \dfn \cbra{x\in \splex \such x^1 \leq x^2  \leq \ldots \leq x^d}.
\end{equation}
For $x \in \splex$, we write $\xord = (x^{(1)},...,x^{(d)})$ for the corresponding ordered point in $\osplex$; also, for $i=1,\ldots,d$, let $r(x^i)$ denote the rank of $x^i$ amongst $x^1, \ldots ,x^d$, with ties resolved in lexicographic order.

As aforementioned, it is natural to assume the vector $\bigxord = (X^{(i)}; \, i=1, \ldots, d)$ of ranked relative capitalizations, which takes values in $\osplex$, is stable in the long run. Thus, we shall take as inputs a pair $(\kappa, q)$ where $\kappa:\osplex\mapsto \mathbb{S}^d_{++}$ and $q:\osplex\mapsto (0,\infty)$ with $\int_{\osplex} q = 1$. Similarly to Definition \ref{defn: Pi_p}, we consider the class of measures $\Pi_{\leq} $ on $C([0,\infty), \splex)$, equipped with the Borel $\sigma$-algebra $\F$, such that for $\prob\in\Pi_{\leq}$:
\begin{enumerate}
	\item $X_t\in \splex$, for all $t\geq 0$, $\prob$-a.s.
	\item $X$ is a $\prob$-semimartingale and, for $i=1,\ldots ,d$ and $j=1,\ldots ,d$:
	\begin{equation}\label{eq:quad_var}
	[X^i,X^j] = \int_0^\cdot \kappa^{r(X^i_t)r(X^j_t)} \big( \bigxord_t \big) d t;\qquad \prob \textrm{-a.s.}
	\end{equation}
	\item For all Borel measurable functions $h$ on $\osplex$ with $\int_{\osplex} h^+ q < \infty$, it holds that
	\begin{equation*}
	\lim_{T\uparrow\infty} \frac{1}{T}\int_0^T h \big(\bigxord_t \big) dt = \int_{\osplex} h q;\qquad \prob\textrm{-a.s.}
	\end{equation*}
	\item The laws of $\{X_t; t \geq 0 \}$ under $\prob$ are tight, where compact sets are those compactly contained within $\splex$.
\end{enumerate}

\begin{rem}\label{rem:pos_def_simplex}

We pause here to discuss two issues arising from the degeneracy of $\splex$.  First, note that we can identify $\splex$ with a region $E \subseteq \reals^{d-1}$ which satisfies Assumption \ref{ass: Ecp}(1), by replacing $x^d$ with $1 - x^1 - \ldots - x^{d-1}$.  However, since $\splex$ is a flat $(d-1)$-dimensional manifold, and for ease of notation, we will not work with $E$, preferring to work directly with $\splex$.

Second, in view of the $(d-1)$-dimensionality of $\splex$, one has to appropriately understand our assumption that $\kappa(\xord)\in\mathbb{S}^d_{++}$ for $\xord\in\osplex$. In fact, $z'\kappa(\xord)z \geq \ell(\xord)z'z$ for some Borel function $\ell : \osplex \mapsto (0, \infty)$ need not hold on the whole of $\Real^d$, but rather on the $(d-1)$-dimensional subspace
\[
\cbra{z \in \Real^d \ \Big| \ \sum_{i=1}^d z_i = 0},
\]
which (up to an affine translation) is tangent to every point $x \in \splex$. Despite the clear abuse of notation, we prefer to write $\kappa(x)\in\mathbb{S}^d_{++}$, understanding that it need only hold on the tangent space to $\splex$.
\end{rem}

\begin{rem}\label{rem: quad_var}
Regarding \eqref{eq:quad_var} above, it may seem more natural to require
\begin{equation}\label{eq:quad_var_ord}
[\bigxord,\bigxord]_{\cdot} = \int_0^\cdot \kappa(\bigxord_t) dt.
\end{equation}
Indeed, as can be deduced from  \cite[Theorem 2.3]{MR2428716}, \eqref{eq:quad_var_ord} implies $X$ has instantaneous covariations
\[
\frac{d [X^i, X^j]_t}{d t} = \kappa^{r(X^i_t)r(X^j_t)}(\bigxord_t); \quad i =1, \ldots, d, \  j =1, \ldots, d, \qquad t \geq 0,
\]
when $\bigxord$ is in the interior of $\osplex$. However, without additional assumptions (for example, almost surely zero Lebesgue measure for ranks coinciding), one cannot assume that $\kappa$ is non-degenerate, or even recover $(d [X, X]_t / d t; \, t \geq 0)$ from $\kappa$ on the boundary of $\osplex$.  For this reason, we define $\Pi_{\leq}$ using \eqref{eq:quad_var} rather than \eqref{eq:quad_var_ord}. Morally, we regard the two definitions as equivalent.
\end{rem}

\subsection{Growth in rank-based models}

In accordance to Definition \ref{defn: calV}, let $\V_{\leq}$ be the class of predictable process $\vartheta$ that are $X$-integrable respect to every $\prob \in \Pi_{\leq}$. Growth rates $\probgrowth{V^{\vartheta}}{\prob}$ for $\vartheta \in \V_{\leq}$ are defined as in \eqref{eq: growth}, and we set
\begin{equation}\label{eq: robust_problem_rank}
\lambda_{\leq} \dfn \sup_{\vartheta\in\V_{\leq}} \inf_{\prob\in\Pi_{\leq}} \probgrowth{V^{\vartheta}}{\prob}.
\end{equation}

We wish to use the results of Section \ref{S:APS} in the current setting.
Of course, one really invests in the relative capitalizations $X$, and not in it's ranked counterpart $X^{()}$; therefore, a transformation of the inputs $(\kappa, q)$ from $\osplex$ to $(c, p)$ in $\splex$ is in order. Additionally, we must ensure that $(c,p)$ satisfy the regularity and integrability requirements of Assumptions \ref{ass: Ecp} and \ref{ass: Ecp_i}.

In view of \eqref{eq:quad_var}, we first naturally extend $\kappa$ from $\osplex$ to $\splex$ by defining
\begin{equation}\label{eq: c_ranked}
c(x) = \cbra{c^{ij}(x)}_{i,j=1}^d; \qquad c^{ij}(x) \dfn \kappa^{r(x^i)r(x^j)}(\xord);\qquad x\in\splex.
\end{equation}
The extension from of $q$ from $\osplex$ to $\splex$ is performed in a ``symmetric'' way, by defining
\begin{equation}\label{eq: p_ranked}
p(x) \dfn \frac{1}{d!}q(\xord);\qquad x\in\splex.
\end{equation}
We make the following assumption (see Remark \ref{rem:pos_def_simplex} regarding part (2) of Assumption \ref{ass: Ecp}):

\begin{ass}\label{ass: kappa_q_extend}
The pair $(\kappa,q)$ is such that, with $E=\splex$, $c$ as in \eqref{eq: c_ranked} and $p$ as in \eqref{eq: p_ranked}, Assumptions \ref{ass: Ecp} and \ref{ass: Ecp_i} are satisfied.
\end{ass}

\begin{rem}
We stress that the passage from $(\kappa, q)$ to $(c, p)$ is made only in order to use the results from Section \ref{S:APS} under the validity of Assumption \ref{ass: kappa_q_extend}. In particular, we do \emph{not} require that models have limiting stable distribution $p$ for $X$, where all rankings are equally likely, each with probability $1 / d!$. Note, however, that the worst-case model will have this structure.
\end{rem}

Under the force of Assumption \ref{ass: kappa_q_extend} on $(\kappa,q)$, build $(c,p)$ as in \eqref{eq: c_ranked} and \eqref{eq: p_ranked}, respectively, as well as $\Pi,\V$ according to Definitions \ref{defn: Pi_p} and \ref{defn: calV}.   It is immediate that
\begin{equation*}
\Pi\subseteq \Pi_{\leq};\qquad \V_{\leq}\subseteq \V.
\end{equation*}
Thus, we always have
\begin{equation}\label{eq: ranked_abs_ub}
\lambda_{\leq} = \sup_{\vartheta \in \V_{\leq}}\inf_{\prob \in \Pi_{\leq}} \probgrowth{V^{\vartheta}}{\prob} \leq \sup_{\vartheta \in  \V}\inf_{\prob \in \Pi} \probgrowth{V^{\vartheta}}{\prob} = \lambda.
\end{equation}
From Theorem \ref{thm: mr}, we know that for $\hat{u}$ solving \eqref{eq:minimization} and $\hat{\vartheta}$ defined in \eqref{eq:trading_strategy}, we have \eqref{eq: lambda_I_equiv} holding for all $\prob\in\Pi$. The following result, the proof of which can be found in Appendix \ref{S:proofs_from_rank}, implies that the portfolio generating function $\hat{u}$ depends only on the ordering of its input coordinates (which is the same as saying that it is permutation invariant), giving also that $\lambda_{\leq}  = \lambda $.

\begin{prop}\label{prop: growth_same}
For the pair $(\kappa, q)$, let Assumption \ref{ass: kappa_q_extend} hold. For the associated $(c,p)$ constructed above, let $\hat{u}$ be as in \eqref{eq:minimization}, from Theorem \ref{thm: mr}. Then, $\hat{u} (x) = \hat{u} (\xord)$ for all $x \in \splex$. Furthermore,
\begin{equation}\label{eq: functionals}
\lambda_{\leq} = \frac{1}{2} \int_{\osplex} q\pare{\frac{\nabla \hat{u}}{\hat{u}}}'\kappa\pare{\frac{\nabla \hat{u}}{\hat{u}}} = \frac{1}{2} \int_{\splex} p \pare{\frac{\nabla \hat{u}}{\hat{u}}}' c \pare{\frac{\nabla \hat{u}}{\hat{u}}} =\lambda,
\end{equation}
and the trading strategy $\hat{\vartheta}_\cdot = (\nabla\hat{u} / \hat{u})(X_\cdot)\in\V_{\leq}$ is such that
\begin{equation}\label{eq: ranked_probgrowth}
\probgrowth{V^{\hat{\vartheta}}}{\prob} = \lambda_{\leq}, \quad \forall \ \prob\in\Pi_{\leq}.
\end{equation}
\end{prop}

\begin{rem}\label{rem:port_near_crossings}
The importance of $\hat{u}$ being a function of the ranked weights is that it implies the optimal strategy is rank-generated, in the sense of \cite[Section 4.2]{MR1894767}. Because $\hat{u}$ is permutation-invariant and twice continuously differentiable, the local time terms in \cite[Theorem 4.2.1]{MR1894767} vanish. In effect, we have $\set{X^i = X^j} \subseteq \{ \hat{\vartheta}^i = \hat{\vartheta}^j \}$ for all $i=1, \ldots, d$ and $j=1, \ldots, d$ in our result, which is a desirable feature. Indeed, if optimal positions were different when the ranks of two stocks are the same, then, upon collisions of ranked market capitalizations, one would need to change large positions with very high frequency. Not only is this practically infeasible, it also would lead to unsustainable transaction costs (which, admittedly, we do not model here).
\end{rem}

We close this section with an important observation. It is not hard to see that Assumption \ref{ass: kappa_q_extend} only concerns $(\kappa,q)$ near $\partial\osplex$.  However, and especially in view of the non-explosivity  of $X^R$ in Assumption \ref{ass: Ecp_i}(iii), it is natural to wonder if there is \emph{ever} a way to modify an arbitrary $(\kappa,q)$ near $\partial\osplex$ so that Assumption \ref{ass: kappa_q_extend} holds. To this end we have the following result.

\begin{prop}\label{prop: modify}
Let $\kappa: \osplex \mapsto \mathbb{S}^d_{++}$ and $q:  \osplex \mapsto (0, \infty)$ be such that for all open subsets $V \subset \osplex$ we have  $\kappa\in C^{2,\gamma}(\bar{V};\mathbb{S}^d_{++})$ and $q\in C^{1,\gamma}(\bar{V};(0,\infty))$.  Then, for any open subset $V \subset \osplex$ there are $(\kappa_V,q_V)$ such that
\begin{enumerate}
	\item[(i)] $q_V$ is strictly positive in $\osplex$ with $\int_{\osplex} q_V(x)dx = 1$;
	\item[(ii)] $\kappa = \kappa_V$, $q=q_V$ on $V$;
	\item[(iii)] $(\kappa_V,q_V)$ satisfy Assumption \ref{ass: kappa_q_extend}.
\end{enumerate}
\end{prop}

In fact, $(\kappa_V, q_V)$ in Proposition \ref{prop: modify} admit explicit formulas in \eqref{eq: qv_kappav}. Upon inspection of Lemma \ref{lem: nice_matrix}, the modified pair $(\kappa_V,q_V)$ is such that, for some constant $K>0$, the optimizer $\hat{u}(x) = \left(\prod_{i=1}^d x^i\right)^K$ for $\xord$ lying near $\partial \osplex$. This leads to an optimal strategy $\hat{\vartheta}$ such that $\hat{\vartheta}^i = K/X^i + (1-Kd)$, $i=1,\dots,d$, when $\bigxord$ is near $\partial\osplex$.  Qualitatively, the investor holds a combination of the market and equally weighted portfolios.  This, of course, is entirely consistent with the set inclusion $\set{X^i = X^j} \subseteq \{ \hat{\vartheta}^i = \hat{\vartheta}^j \}$ for $i,j=1,\ldots,d$ of Remark \ref{rem:port_near_crossings}.

\appendix

\section{Proof of Theorem \ref{thm: mr}}
\label{appsec:proof_of_main}
We first provide a brief road-map on how Theorem \ref{thm: mr} is proved, starting with the variational problem \eqref{eq:minimization}. Consider when $u = e^{(1/2)\phi}$ for $\phi \in C^\infty_c(E)$. Clearly, $u\in\D$, and from \eqref{eq: v_u} we deduce that for all $\prob\in\Pi$
\begin{equation*}
\probgrowth{V^{\vartheta^u}}{\prob} = -\frac{1}{8}\int_E \left(\nabla\phi'c\nabla\phi + 2\tr\pare{cD^2\phi}\right) p.
\end{equation*}
As $\phi \in C^\infty_c(E)$, integration-by-parts yields
\begin{equation*}
\begin{split}
\probgrowth{V^{\vartheta^u}}{\prob} &= -\frac{1}{8}\int_E \left(\nabla\phi'c\nabla\phi - 2\nabla\phi'c\left(\frac{\nabla p}{p} + c^{-1}\matdiv{c}\right)\right)p;\\
&= \frac{1}{8}\int_E \ell'c\ell p - \frac{1}{8}\int_E \left(\nabla\phi - \ell \right)'c\left(\nabla\phi - \ell\right)p,
\end{split}
\end{equation*}
where $\ell$ is from \eqref{eq: ell_def}. Thus, we conjecture that $\hat{u}$ in \eqref{eq:minimization} is found by solving
\begin{equation}\label{eq: min_prob_h}
\inf_{\phi} \int_E \left(\nabla \phi - \ell\right)'c\left(\nabla\phi - \ell\right)p,
\end{equation}
and setting $\hat{u} = e^{(1/2)\hat{\phi}}$ if a minimizer exists. Of course, since we actually need the minimizer, we cannot take the infimum over $C^\infty_c(E)$. Instead we use $W^{1,2}_{\loct}(E)$, the space of weakly differentiable functions $\phi$ so that $\phi^2,|\nabla\phi|^2$ are locally integrable. The first result we shall provide, Lemma \ref{lem: min_val} in Subsection \ref{A:analytic_proof}, identifies a unique (up to an additive constant) minimizer $\hat{\phi}\in W^{1,2}_{\loct}(E)$, which is in fact twice continuously differentiable with H\"older second-order derivative.

Given a minimizer $\hat{\phi}$, the first order condition for optimality in the minimization problem of \eqref{eq: min_prob_h} suggests that
\begin{equation} 
\nabla\cdot\left(pc\left(\nabla\hat{\phi} - \ell \right)\right) = 0.
\end{equation}
This is indeed shown to hold in Lemma \ref{lem: min_val}. Therefore, \cite[Corollary 4.9.4]{MR1326606} implies that, \emph{if} $\hat{X}$ from \eqref{eq: p_star_diffusion} does not explode, then it is ergodic and $p$ is the invariant measure. Therefore, the second result, Lemma \ref{lem: ergodic} is Subsection \ref{A:prob_proof}, will establish the fact that $\hat{X}$ from \eqref{eq: p_star_diffusion} does \emph{not} explode.

Given the previous two auxiliary results, Subsection \ref{S:Main_Proof} will conclude the proof of Theorem \ref{thm: mr}.

\subsection{The variational problem}\label{A:analytic_proof}

We first consider the minimization problem in \eqref{eq: min_prob_h} and obtain the following result.

\begin{lem}\label{lem: min_val}
	
Let Assumption \ref{ass: Ecp} and Assumption \ref{ass: Ecp_i}(i) hold, and recall $\ell$ from \eqref{eq: ell_def}. Then, there exists a unique (up to an additive constant) $\hat{\phi}\in W^{1,2}_{\loct}(E)$ which solves
	\begin{equation}\label{eq: min_val}
	\inf_{\phi\in W^{1,2}_{\loct}(E)} \int_{E} \left(\nabla\phi -\ell\right)'c\left(\nabla\phi -\ell \right)p.
	\end{equation}
	Furthermore, $\hat{\phi}\in C^{2,\gamma'}(E)$ for some $0<\gamma'\leq \gamma$ and satisfies the second order linear elliptic equation
	\begin{equation}\label{eq: diff_val}
	\nabla\cdot\left(pc\left(\nabla\hat{\phi} - \ell\right)\right) = 0;\qquad x\in E.
	\end{equation}
\end{lem}

\begin{proof}
To make the notation cleaner set
\begin{equation}\label{eq: ell_J_def}
\begin{split}
J(\phi) \dfn \int_E (\nabla\phi-\ell)'c(\nabla\phi - \ell)p,
\end{split}
\end{equation}
so that \eqref{eq: min_val} becomes $\hat{J}\dfn \inf_{\phi\in W^{1,2}_{\loct}(E)} J(\phi)$. Note that Assumption \ref{ass: Ecp_i}(i) gives $\hat{J}<\infty$. In what follows $K$ will be a constant which changes from line to line. Also, where appropriate, $K_n$ will be a constant which depends only upon $E_n$ and the model coefficients on $E_n$.

Let $\cbra{\phi_m}_{m\in\mathbb{N}}\subset W^{1,2}_{\loct}(E)$ be such that $\lim_{m\uparrow\infty} J(\phi_m) = \hat{J}$. Assumption \ref{ass: Ecp_i}(i) and the Cauchy-Schwarz inequality then imply
\begin{equation*}
\sup_{m}\int_{E} \nabla\phi_m'c\nabla\phi_m p \leq K,
\end{equation*}
and hence for all $n$
\begin{equation*}
\sup_{m}\int_{E_n} \nabla\phi_m'c\nabla\phi_m p \leq K.
\end{equation*}
Next, since $p\geq c_n> 0$ and $c\geq \lambda_n >0$ on $E_n$ we have that
\begin{equation}\label{eq: phi_m_n_grad}
\sup_{m}\int_{E_n} \nabla\phi_m'\nabla\phi_m \leq K_n.
\end{equation}
Denote by
\begin{equation*}
\psi^n_m \dfn \phi_m - \oint_{E_n}\phi_m,
\end{equation*}
as $\phi_m$ less it's average over $E_n$. From the classical Poincare inequality \cite[Chapter 5.8]{MR1625845} and \eqref{eq: phi_m_n_grad} it follows that
\begin{equation*}
\sup_m \int_{E_n}\left((\psi^n_m)^2 + (\nabla\psi^n_m)'\nabla\psi^n_m\right) \leq K_n.
\end{equation*}
The Rellich-Kondrachov theorem \cite[Chapter 5.7]{MR1625845} and the fact that $\nabla\psi^n_m$ is norm bounded in $L^2(E_n,\reals^d)$ imply the existence of $\eta^n\in W^{1,2}(E_n)$ such that for some subsequence $m(n)$:
\begin{equation*}
\begin{split}
\psi^n_{m(n)} &\rightarrow \eta^n \textrm{ strongly in } L^2(E_n),\\
\nabla\psi^n_{m(n)}&\rightarrow \nabla\eta^n \textrm{ weakly in } L^2(E_n,\reals^d).
\end{split}
\end{equation*}
Thus, by \cite[Theorem 13.1.1]{MR2192832}, which shows that
\begin{equation}\label{eq: wlsc}
L^2(E_n;\reals^d)\ni v \mapsto \int_{E_n} (v-\ell)'c(v-\ell)p,
\end{equation}
is weakly lower-semicontinuous  it follows that
\begin{equation}\label{eq: int_form_n}
\int_{E_n}\left(\nabla\eta^n-\ell\right)'c\left(\nabla\eta^n-\ell\right)p \leq \liminf_{m(n)\rightarrow\infty}\int_{E_n}\left(\nabla \phi_m - \ell\right)'c\left(\nabla\phi_m-\ell\right)p.
\end{equation}
Now, fix $n < n'$.  There exists a common subsequence $m(n,n')$ such that
\begin{equation*}
\begin{split}
\phi_{m(n,n')} &- \oint_{E_n}\phi_{m(n,n')} \rightarrow \eta^n;\qquad \textrm{s- }L^2(E_n), \textrm{w- }W^{1,2}(E_n),\\
\phi_{m(n,n')} &- \oint_{E_{n'}}\phi_{m(n,n')}\rightarrow \eta^{n'};\qquad \textrm{s- }L^2(E_{n'}), \textrm{w- }W^{1,2}(E_{n'}),
\end{split}
\end{equation*}
(we have used ``s'' and ``w'' to denote strong and weak convergence). We now claim that $\nabla \eta^n = \nabla \eta^{n'}$ a.e. in $E_n$. Indeed, we have for all $v\in L^2(E_n;\reals^d)$ that
\begin{equation*}
\int_{E_n}\left(\nabla\eta^n - \nabla\eta^{n'}\right)'v = \lim_{m(n,n')\rightarrow\infty}\int_{E_n}\left(\nabla\phi_{m(n,n')} - \nabla\phi_{m(n,n')}\right)'v = 0,
\end{equation*}
upon which the result follows by taking $v = \nabla\eta^n - \nabla\eta^{n'}$.  Thus, since $E_n$ is connected we know \cite[Chapter 5]{MR1625845} that for some constant $C(n,n')$
\begin{equation}\label{eq: eta_diff_n_const}
\begin{split}
\eta^{n'} &= \eta^n + C(n,n');\qquad a.e. \textrm{ in } E_n,\\
\nabla \eta^{n'} &= \nabla\eta^n;\qquad a.e. \textrm{ in } E_n.
\end{split}
\end{equation}
Now, using the double-subsequence trick we can find a single subsequence (which we will label $m$) such that the above convergences holds for all $n\in\mathbb{N}$. For this subsequence (and the resultant $\eta^n$) define (for a.e. $x\in E$) $v$ by
\begin{equation}\label{eq: v_def}
v(x) \dfn \nabla\eta^n(x);\qquad x\in E_n;n= 1,2,\dots.
\end{equation}
Note that $v$ is well defined: indeed we have
\begin{equation*}
\begin{split}
v(x) &= \nabla\eta^1(x) = \nabla\eta^2(x) = \ldots;\qquad x\in E_1,\\
v(x) &= \nabla\eta^2(x) = \nabla\eta^3(x) = \ldots;\qquad x\in E_2,\\
&\vdots.
\end{split}
\end{equation*}
Next, define (for a.e. $x\in E$) $\eta$ by
\begin{equation}\label{eq: eta_def}
\eta(x) \dfn \eta^n(x) - \sum_{k=1}^{n-1}c(k,k+1);\qquad x\in E_n;\qquad n=1,2,\dots.
\end{equation}
Again, $\eta$ is well defined. This follows because for any $n=1,2,\dots$ and $q=0,1,2,\dots$ we have on $E_n$ that
\begin{equation*}
\begin{split}
\eta^{n+q}(x) - \sum_{k=1}^{n+q-1}c(k,k+1) &= \eta^{n+q-1}(x) + c(n+q-1,n+q) - \sum_{k=1}^{n+q-1}c(k,k+1),\\
&= \eta^{n+1-q}(x) - \sum_{k=1}^{n+q-2}c(k,k+1),\\
&\vdots\\
&= \eta^n(x) - \sum_{k=1}^{n-1}c(k,k+1).
\end{split}
\end{equation*}
We now claim that $\eta\in W^{1,2}_{\loct}(E)$. First $\nabla \eta = v$. To see this, let $\theta\in C^{\infty}_c(E)$. Choose $n$ so that $\eta\in C^{\infty}_c(E_{n})$. For $i\in 1,\ldots,d$ write $D^i$ as the derivative with respect to $x_i$. We have
\begin{equation*}
\begin{split}
\int_E \eta D^{i}\theta &= \int_{E_{n}}\eta D^i\theta = \int_{E_n}\left(\eta^n - \sum_{k=1}^{n-1}c(k,k+1)\right)D^{i}\theta;\\
&=-\int_{E_n}D^i\eta^n \theta =-\int_{E_n}v^i\theta.
\end{split}
\end{equation*}
Given that $\nabla\eta = v$ the fact that $\eta\in W^{1,2}_{\loct}(E)$ is immediate. From \eqref{eq: int_form_n} we thus have for each $n$ that
\begin{equation*}
\begin{split}
\int_{E_n}\left(\nabla\eta - \ell\right)'c\left(\nabla\eta-\ell\right)p &\leq \liminf_{m\rightarrow\infty}\int_{E_n}\left(\nabla\phi_m-\ell\right)'c\left(\nabla\phi_m-\ell\right)p,\\
&\leq \liminf_{m\rightarrow\infty}\int_{E}\left(\nabla\phi_m-\ell\right)'c\left(\nabla\phi_m-\ell\right)p = \hat{J}.
\end{split}
\end{equation*}
Taking $n\uparrow\infty$ and using the monotone convergence theorem we see that
\begin{equation*}
\begin{split}
\int_{E}\left(\nabla\eta - \ell\right)'c\left(\nabla\eta-\ell\right)p &\leq \hat{J},
\end{split}
\end{equation*}
and hence $\hat{\phi}\dfn\eta$ minimizes $J$ over $W^{1,2}_{\loct}(E)$.  The uniqueness up to an additive constant follows by the strict convexity of $(\nabla \phi - \ell)'c(\nabla\phi-\ell)p$ in $\nabla\phi$.

We now turn to the regularity for $\hat{\phi}$ which essentially is a standard argument and hence just a broad overview is given. Let $\theta\in C^1_c(E_n) \subset W^{1,2}_{\loct}(E)$. By varying $J$ at $\hat{\phi}\pm\eps\theta$ and taking $\eps\downarrow 0$ we see that
\begin{equation}\label{eq: loc_var_reg}
0 = \int_{E_n}\nabla\theta'c\left(\nabla\hat{\phi}- \ell\right)p.
\end{equation}
It thus follows \cite[Chapter 8, pp 178]{MR1814364}, that $u=\hat{\phi}$ is a weak solution of the PDE
\begin{equation}
\begin{split}
\nabla\cdot\left(pc\nabla u - pc\ell\right)&= 0;\qquad x\in E_n\\
u&= \hat{\phi};\qquad x\in\partial E_n.
\end{split}
\end{equation}
Here, the boundary condition is interpreted to mean that $u-\hat{\phi} \in W^{1,2}_0(E_n)$. Under the given regularity and ellipticity assumptions in $E_n$ it follows by \cite[Theorem 8.22]{MR1814364} that $u=\hat{\phi}$ is locally Holder continuous in $E_n$ for some exponent $0 < \gamma'\leq \gamma$.  Next, consider the problem of finding classical solutions to the same PDE but in $E_{n-1}$: i.e.
\begin{equation*}
\begin{split}
\nabla\cdot\left(pc\nabla u - pc\ell\right)&= 0;\qquad x \in E_{n-1};\\
u&= \hat{\phi};\qquad x\in\partial E_{n-1}.
\end{split}
\end{equation*}
Since $\hat{\phi}$ is Holder continuous in $\bar{E}_{n-1}$ it follows from \cite[Theorem 6.13]{MR1814364} that there is a unique solution $u\in C^{2,\gamma'}(E_{n-1})\cap C(\bar{E}_{n-1})$ to the above PDE.  But, this means that $u$ is a weak solution to the above PDE as well.  By the uniqueness of weak solutions (which also holds in the current setup - see \cite[Theorem 8.3]{MR1814364} and the fact that $\hat{\phi}$ is also a weak solution in $E_{n-1}$) it follows that $\hat{\phi} = u$ a.e. in $E_{n-1}$. Since we already know $\hat{\phi}$ is holder continuous, this in fact proves that $\hat{\phi}\in C^{2,\gamma}(E_{n-1})$ and solves the differential expression in \eqref{eq: diff_val} in $E_{n-1}$.  Since this works for any $n$ the result follows.
\end{proof}

\subsection{An ergodic diffusion}\label{A:prob_proof}

Having established existence of minimizer $\hat{\phi}$ we now consider the diffusion as in \eqref{eq: p_star_diffusion} with $\hat{u} = e^{(1/2)\hat{\phi}}$: i.e.
\begin{equation}\label{eq: diffusion_hat_phi}
d\hat{X}_t = c\frac{\nabla\hat{u}}{\hat{u}}(\hat{X}_t)dt + \sigma(\hat{X}_t)d\hat{W}_t = \frac{1}{2}c\nabla\hat{\phi}(\hat{X}_t)dt + \sigma(\hat{X}_t)dW_t,
\end{equation}
Our goal is to show $\hat{X}$ is ergodic with invariant measure $p$. More precisely, we let $\hat{\prob} = (\hat{\prob}^x)_{x\in E}$ denote the solution to the generalized martingale problem for the second order linear operator $\hat{L}$ associated to $\hat{X}$ on $E$: i.e.
\begin{equation}\label{eq: hat_L_def}
\hat{L} \dfn \frac{1}{2}\tr\pare{cD^2} + \frac{1}{2}\nabla\hat{\phi}'c\nabla.
\end{equation}
We then have the following which proves $\hat{\prob}\in\tilde{\Pi}$.

\begin{lem}\label{lem: ergodic}  Let Assumption \ref{ass: Ecp} and Assumptions \ref{ass: Ecp_i} (i) and (iii) hold. Let $\hat{\phi}$ be as in Lemma \ref{lem: min_val}. Set $\hat{\prob}$ as the solution to the generalized martingale problem for the operator $\hat{L}$ in \eqref{eq: hat_L_def}. Then $\hat{\prob}$ solves the martingale problem for $\hat{L}$ and in fact, $\hat{\prob}\in\Pi$.
\end{lem}

The rest of this subsection is devoted to the proof of Lemma \ref{lem: ergodic}. We retain the notation of \eqref{eq: ell_def}, \eqref{eq: ell_J_def}.  Since $\hat{\phi}\in C^{2,\gamma'}(E)$ and solves \eqref{eq: diff_val} it follows that
\begin{equation*}
\nabla\cdot\left(p\left(\frac{1}{2}c\nabla\hat{\phi}\right)\right) = \frac{1}{2}\nabla\cdot\left(c\nabla p + p\matdiv{c}\right),
\end{equation*}
and hence $p$ is an invariant density for $\hat{X}$.  Since $\int_E p = 1$ it will follow that $\hat{X}$ is ergodic with invariant measure $p$ if it can be shown that $\hat{X}$ is recurrent in $E$: see \cite[Theorem 4.9.5]{MR1326606}. To this end, we use the results of \cite[Section 6.6]{MR1326606} which provide necessary and sufficient conditions for $\hat{X}$ to be recurrent in the current setup.

We first state a consequence of Assumption \ref{ass: Ecp_i}(iii). Denote by $E_n - E_1\dfn E_n\setminus \bar{E}^c_{1}$ and $L^R$ the second order linear elliptic operator associated to the diffusion $X^R$ from Assumption \ref{ass: Ecp_i}(iii): i.e. in divergence form
\begin{equation}\label{eq: L_R_def_div}
L^R = \frac{1}{2}\nabla\cdot\left(c\nabla\right) + \frac{1}{2}\frac{\nabla p}{p}'c\nabla.
\end{equation}
Since $X^R$ s assumed ergodic with invariant density $p$ and reversing by construction, it follows from \cite[Theorem 6.4.1]{MR1326606} that
\begin{equation}\label{eq: u_R_integral}
\lim_{n\uparrow\infty} \frac{1}{2}\int_{E_n - E_1}(\nabla u_n)'c\nabla u_n p = 0,
\end{equation}
where $u_n\in C^{2,\gamma}(E_n - E_1)$ is the unique (strictly positive in $E_n - E_1$) solution to
\begin{equation}\label{eq: u_n_pde}
L^R u = 0,\ x\in E_n - E_1;\qquad u = 1,\ x\in\partial E_1;\qquad u = 0, \ x\in \partial E_n.
\end{equation}
In fact, one has
\begin{equation}\label{eq: u_n_prob_rep}
u_n(x) = \prob^R_x\bra{\tau_{E_1} < \tau_{E_n}},
\end{equation}
where $\cbra{\prob^R_x}_{x\in E}$ is the solution to the Martingale problem for $L^R$ on $E$ and $\tau_{E_i}$ is the first hitting time to $\partial E_i$.  Note that this implies $0\leq u^R\leq 1$. To show that $X$ is recurrent we use the following result, as can be found in \cite[Theorem 6.6.1]{MR1326606}:

\begin{thm}\label{T:pinsky}
Let Assumptions \ref{ass: Ecp} -- \ref{ass: Ecp_i}(i) hold and let $\hat{\phi}$ be as in Lemma \ref{lem: min_val}. Let $\hat{L}$ be operator associated to $\hat{X}$ in \eqref{eq: diffusion_hat_phi}. For each $n$ define the convex sets
\begin{align*}
A_n &\dfn \cbra{g\in W^{1,2}(E_n - E_1): g = \sqrt{p} \textrm{ on } \partial E_1, g = 0\textrm{ on } \partial E_n, \textrm{dist}(x,\partial E_n)^{-1}g(x)\in L^\infty(E_n - E_1)};\\
B_n &\dfn \cbra{h \in W^{1,2}(E_n - E_1,g^2): h = \frac{1}{2}\log(p)\textrm{ on } \partial E_1}.
\end{align*}
Now, define
\begin{equation*}
\begin{split}
\mu_n &\dfn \inf_{g\in A_n}\sup_{h\in B_n}\frac{1}{2}\int_{E_n - E_1} g^2\left(\frac{\nabla g}{g} - \frac{1}{2}\left(\nabla \hat{\phi} - c^{-1}\matdiv{c}\right)\right)'c\left(\frac{\nabla g}{g} - \frac{1}{2}\left(\nabla \hat{\phi} - c^{-1}\matdiv{c}\right)\right)\\
&\qquad\qquad - \frac{1}{2}\int_{E_n - E_1} g^2\left(\nabla h - \frac{1}{2}\left(\nabla\hat{\phi} - c^{-1}\matdiv{c}\right)\right)'c\left(\nabla h - \frac{1}{2}\left(\nabla\hat{\phi} - c^{-1}\matdiv{c}\right)\right).
\end{split}
\end{equation*}
Then, $\hat{L}$ is recurrent if and only if $\lim_{n\uparrow\infty}\mu_n = 0$.

\end{thm}

\begin{rem} Above, $W^{1,2}(E_n - E_1,g^2)$ is the space of weakly differentiable functions $h$ satisfying
\begin{equation*}
\int_{E_n - E_1}g^2\left(h^2 + \nabla h'\nabla h\right) < \infty.
\end{equation*}
Also, the boundary conditions are interpreted to hold in the trace sense. Lastly, as shown right above \cite[Theorem 6.6.1]{MR1326606}, $\mu_n$ takes the simpler form
\begin{equation}\label{eq: mu_n_alt}
\mu_n = \frac{1}{2}\int_{E_n - E_1}\frac{\tilde{v}_n}{v_n}\nabla v_n'c\nabla v_n,
\end{equation}
where $v_n$ solves $\hat{L}v = 0$ in $E_n-\bar{E}_1$ with $v = 0$ on $\partial E_n$ and $v =1$ on $\partial E_1$; and, with $\tilde{L}$ denoting the formal adjoint to $\hat{L}$, where $\tilde{v}_n$ solves $\tilde{L}\tilde{v} = 0$ on $E_n-\bar{E}_1$ with $\tilde{v} = 0$ on $\partial E_n$ and $1$ on $\partial E_1$.  Thus, if $\cbra{\hat{\prob}_x}_{x\in E}$ denotes the solution to the generalized Martingale problem for $\hat{L}$ on $E$ then $v_n(x) = \hat{\prob}_x\bra{\tau_{E_1}<\tau_{E_n}}$ where $\tau_{E_i}$ is the first hitting time of $\partial E_i$. Thus, if  $\mu_n\rightarrow 0$ then $v_n\rightarrow 1$ so that $\hat{L}$ is non-explosive, hence positive recurrent since $p$ is in invariant probability density. If $\mu_n\not\rightarrow 0$ then $v_n\not\rightarrow 1$ and accordingly $L$ is transient. This is the idea behind the proof.
\end{rem}

Now, \eqref{eq: mu_n_alt} implies $\mu_n\geq 0$ and hence $\liminf_{n\uparrow\infty}\mu_n\geq 0$. Assume by way of contradiction that $\limsup_{n\uparrow\infty}\mu_n > 0$. Thus, for some $\delta >0$ we have $\mu_n \geq \delta$ for all $n$.  Taking $g \dfn u_n\sqrt{p}$ (which by the global Schauder estimates for $u_n$ is in $A_n$: see \cite[Theorem 3.2.8]{MR1326606}) one obtains
\begin{align*}
\delta &\leq \frac{1}{8}\int_{E_n - E_1} u_n^2\left(\nabla\hat{\phi} - \left(\frac{\nabla p}{p} + c^{-1}\matdiv{c}\right) -  2\frac{\nabla u_n}{u_n}\right)'c\left(\nabla\hat{\phi} - \left(\frac{\nabla p}{p} + c^{-1}\matdiv{c}\right) -  2\frac{\nabla u_n}{u_n}\right)p\\
&\qquad\qquad - \frac{1}{2} \inf_{h\in B_n} \int_{E_n - E_1} u_n^2\left(\nabla h - \frac{1}{2}\left(\nabla \hat{\phi} - c^{-1}\matdiv{c}\right)\right)'c\left(\nabla h - \frac{1}{2}\left(\nabla \hat{\phi} - c^{-1}\matdiv{c}\right)\right)p; \\
&= \frac{1}{8}\int_{E_n - E_1} u_n^2\left(\nabla\hat{\phi} - \ell -  2\frac{\nabla u_n}{u_n}\right)'c\left(\nabla\hat{\phi} - \ell -  2\frac{\nabla u_n}{u_n}\right)p\\
&\qquad\qquad - \frac{1}{2} \inf_{h\in B_n} \int_{E_n - E_1} u_n^2\left(\nabla h - \frac{1}{2}\left(\nabla\hat{\phi} - c^{-1}\matdiv{c}\right)\right)'c\left(\nabla h - \frac{1}{2}\left(\nabla\hat{\phi} - c^{-1}\matdiv{c}\right)\right)p,
\end{align*}
where we have used the definition of $\ell$ in \eqref{eq: ell_J_def}. Next, for $h\in B_n$ define $\phi \dfn \log(p) + \hat{\phi} - 2h$. Under the given regularity assumptions on $p,\hat{\phi}$ we have by the linearity of the trace operator that
\begin{equation*}
h \in B_n \Leftrightarrow \phi \in B'_n \dfn \cbra{\phi\in W^{1,2}(E_n - E_1,pu_n^2): \phi = \hat{\phi}\textrm{ on } \partial E_1}.
\end{equation*}
The change of variables $h=(1/2)\left(\log(p)+\hat{\phi}-\phi\right)$ and simple algebra in the previous inequality gives
\begin{align*}
\inf_{\phi\in B'_n}&\int_{E_n - E_1} u_n^2\left(\nabla \phi - \ell\right)'c\left(\nabla \phi - \ell\right)p\\
&\qquad\qquad \leq \int_{E_n - E_1} u_n^2\left(\nabla \hat{\phi} - \ell -  2\frac{\nabla u_n}{u_n}\right)'c\left(\nabla\hat{\phi} - \ell -  2\frac{\nabla u_n}{u_n}\right)p - 8 \delta.
\end{align*}
Since $\hat{\phi}\in W^{1,2}_{\loct}(E)$ and, according to Lemma \ref{lem: min_val} satisfies $\int_E p\nabla\hat{\phi}'c\nabla\hat{\phi} < \infty$, by \eqref{eq: u_R_integral} we know that
\begin{align*}
\lim_{n\uparrow\infty}&\int_{E_n - E_1} u_n^2\left(\nabla \hat{\phi} - \ell -  2\frac{\nabla u_n}{u_n}\right)'c\left(\nabla \hat{\phi} - \ell -  2\frac{\nabla u_n}{u_n}\right)p =  \int_{E - E_1} \left(\nabla\hat{\phi} - \ell\right)'c\left(\nabla\hat{\phi} - \ell\right)p.
\end{align*}
Thus, for $n$ large enough we have
\begin{align*}
\inf_{\phi\in B'_n}&\int_{E_n - E_1} u_n^2\left(\nabla \phi - \ell\right)'c\left(\nabla \phi - \ell\right)p \leq \int_{E - E_1} \left(\nabla \hat{\phi} - \ell\right)'c\left(\nabla\hat{\phi} - \ell\right)p - 4\delta.
\end{align*}
Now, by Assumption \ref{ass: Ecp_i}(i) it follows that $\int_{E_n - E_1}\ell'c\ell p< \infty$. Thus, as shown in \cite[Theorem 6.5.1, pp 264]{MR1326606}, there exists an a.e. unique (up to an additive constant) solution $\phi_n\in W^{1,2}(E_n - E_1,pu_n^2)$ to the minimization problem above. Indeed, to connect with the proof therein take $g = u_n\sqrt{p}$, $f=\hat{\phi}$, $\phi = 1$ on $\partial E_1$ and $\phi=0$ on $\partial E_n$ and lastly $a=c$, $b=\ell$.  Therefore, we have
\begin{equation*}
\begin{split}
\int_{E_n - E_1} u_n^2\left(\nabla \phi_n - \ell\right)'c\left(\nabla \phi_n - \ell\right)p \leq \int_{E - E_1} \left(\nabla \hat{\phi} - \ell\right)'c\left(\nabla\hat{\phi} - \ell\right)p - 4\delta.
\end{split}
\end{equation*}
Next, extend $u_n$ to all of $E_n$ by setting $u_n = 1$ on $E_1$. It is well known (see \cite[Proposition 5.1.1]{MR2192832}) that since $u_n = 1$ on $\partial E_1$ that this extension is in $W^{1,2}(E_n)$ (in fact, it is continuous, though not continuously differentiable because of the Hopf maximum principle). Similarly, for $\phi\in B'_n$ we have $\phi = \hat{\phi}$ on $\partial E_1$ and hence we may extend $\phi$ to $E_n$ by setting $\phi = \hat{\phi}$ in $E_1$ and it still holds that $\ell\in W^{1,2}(E_n,p u_n^2)$. This gives for $n$ large enough, say $n\geq N_0(\delta)$ that
\begin{equation}\label{eq: var_ineq}
\begin{split}
&\int_{E_n} u_n^2\left(\nabla \phi_n - \ell\right)'c\left(\nabla \phi_n - \ell\right)p\leq \int_{E} \left(\nabla \hat{\phi} - \ell\right)'c\left(\nabla \hat{\phi} - \ell\right)p - 4\delta = J(\hat{\phi}) - 4\delta,
\end{split}
\end{equation}
where we recall the definition of $J$ in \eqref{eq: ell_J_def}.  We now use \eqref{eq: var_ineq} to derive a contradiction to Lemma \ref{lem: min_val}.  To do this, fix an integer $m$. Since $Lu_n = 0$, Harnack's inequality, $u_n\leq 1$ (which follows by the probabilistic representation for $u_n$ in \eqref{eq: u_n_prob_rep}) and $u_n(x) = 1$ on $E_1$ yields the existence of a constant $c_m$ so that $u_n(x)^2 \geq c_m$ on $E_m$ for all $n\geq m+1$.\footnote{Technically, Harnack's inequality holds in $E_m - E_2$ where $u_n$ is smooth. The extension to all of $E_m$ follows since by the extension, $u_n = 1$ on $\bar{E}_1$ and since $u_n$ is larger in $E_2 - E_1$ than in $E_m - E_2$, as the probabilistic representation shows.} We thus have for $n\geq N_0(\delta)\vee (m+1)$ that
\begin{equation}\label{eq: j_calc}
\begin{split}
J(\hat{\phi})-4\delta &\geq \int_{E_n}u_n^2\left(\nabla\phi_n-\ell\right)'c\left(\nabla\phi_n-\ell\right)p;\\
&\geq \int_{E_m}u_n^2\left(\nabla\phi_n-\ell\right)'c\left(\nabla\phi_n-\ell\right)p;\\
&\geq c_m\int_{E_m}\left(\nabla\phi_n-\ell\right)'c\left(\nabla\phi_n-\ell\right)p;\\
&\geq c_m\int_{E_m}\left(\nabla\phi_n-\ell\right)'\left(\nabla\phi_n-\ell\right),
\end{split}
\end{equation}
where $c_m$ has changed to the last line, taking into account that $p,c\geq c_m > 0$ on $E_m$. Thus, from the Cauchy-Schwartz inequality we have
\begin{equation*}
\sup_{n\geq N_0(\delta)\vee (m+1)} \int_{E_m}\nabla\phi_n'\nabla\phi_n \leq K_m.
\end{equation*}
Copying the argument below \eqref{eq: phi_m_n_grad} in Lemma \ref{lem: min_val} (note the roles of $m$ and $n$ have reversed) there exists a function $\eta^m \in W^{1,2}(E_m)$ so that, for some subsequence $n(m)$
\begin{align*}
\phi_n - \oint_{E_m}\phi_n &\rightarrow \eta^m;\qquad \textrm{s- } L^2(E_m);\\
\nabla\phi_n &\rightarrow \nabla \eta^m;\qquad \textrm{w- } L^2(E_m;\reals^d).
\end{align*}
Furthermore, if $m<m'$ then by taking a common subsequence $\eta^{m'} = \eta^m + C(m,m')$ and $\nabla\eta^{m'} = \nabla\eta^m$ almost everywhere in $E_m$. In fact, there exist a single subsequence labeled $n$ such that the convergence holds for all $m$ on this subsequence and we can construct a function $\eta\in W^{1,2}_{\loct}(E)$, exactly as in Lemma \ref{lem: min_val}, so that $\nabla\eta = \nabla\eta^m$ on $E_m$ for each $m$.

Now, come back to \eqref{eq: var_ineq}. For the common subsequence $\cbra{\phi_n}_{n\in\mathbb{N}}$ where all the convergences take place, for each $m$ we have for $n\geq m$ that (similarly to \eqref{eq: j_calc})
\begin{equation}\label{eq: j_calc_2}
J(\hat{\phi})-4\delta \geq \int_{E_m}u_n^2\left(\nabla\phi_n-\ell\right)'c\left(\nabla\phi_n-\ell\right)p \geq \inf_{E_m}u_n^2 \int_{E_m}\left(\nabla\phi_n-\ell\right)'c\left(\nabla\phi_n-\ell\right)p.
\end{equation}
We now claim that for each $m$
\begin{equation}\label{eq: u_n_limit}
\lim_{n\uparrow\infty}\inf_{E_m}u_n^2 = 1.
\end{equation}
First of all, for $x\in \bar{E}_1$ we have $u_n(x) = 1$ by construction.  Secondly, in $E_m - E_1$ we have, since $L^Ru_n = 0$, $u_n = 1$ on $\partial E_1$, and $u_n\leq 1$ on $\bar{E_n - E_1}$, from the global Schauder estimates \cite[Theorem 3.2.8]{MR1326606}, there is a constant $K_m$ so that $\sup_{n\geq m} \|u_n\|_{2,\gamma,E_m} \leq K_m$, where $\|\cdot\|_{2,\gamma, E_m}$ is the $C^{2,\gamma}$ H\"older norm on $E_n$.  Now, assume there is some subsequence (still labeled $n$) so that $\lim_{n\uparrow\infty}\inf_{E_m}u_n^2 = 1-\eps$ for some $\eps > 0$. By the Schauder estimates, $\cbra{u_n}_{n\in\mathbb{N}}$ is precompact in the $\|\cdot\|_{2,\gamma,E_m}$ norm and there is a further subsequence (still labeled $n$) and a function $u_\infty\in C^{2,\gamma}(E_m)$ so that $\|u_n-u_{\infty}\|_{2,\gamma,E_m}\rightarrow 0$.  But, from Assumption \ref{ass: Ecp_i}(iii) and \cite[Theorem 6.4.1]{MR1326606}, we a-priori know that $u_n(x)\rightarrow 1$ so that $u_\infty(x) = 1$.  But this gives
\begin{equation*}
0 = \lim_{n\uparrow\infty}\sup_{E_m}|u_n(x) -1| = 0,
\end{equation*}
which contradicts the fact that $\lim_{n\uparrow\infty}\inf_{E_m}u_n(x) = 1-\eps$. Thus, \eqref{eq: u_n_limit} holds.

Now, come back to \eqref{eq: j_calc_2}. In view of \eqref{eq: u_n_limit} and the lower-semiconiuity of the operator in \eqref{eq: wlsc} (with $n$ there-in equal to $m$ here) it follows that
\begin{equation*}
\begin{split}
J(\hat{\phi}) - 4\delta &\geq \liminf_{n\uparrow\infty} \inf_{E_m}u_n^2\int_{E_m}\left(\nabla\phi_n-\ell\right)'c\left(\nabla\phi_n-\ell\right)p;\\
&\geq \int_{E_m}\left(\nabla\eta^m-\ell\right)'c\left(\nabla\eta^m-\ell\right)p;\\
&= \int_{E_m}\left(\nabla\eta - \ell\right)'c\left(\nabla\eta - \ell\right)p.
\end{split}
\end{equation*}
Taking $m\uparrow\infty$ yields
\begin{equation*}
J(\hat{\phi})-4\delta \geq \int_{E}\left(\nabla\eta - \ell\right)'c\left(\nabla\eta - \ell\right)p,
\end{equation*}
contradicting Lemma \ref{lem: min_val}.  Thus, it cannot be that $\limsup_{n\uparrow\infty}\mu_n > 0$ and hence $\lim_{n\uparrow\infty} \mu_n = 0$, proving the recurrence of $\hat{X}$.

\qed

\subsection{Proof of Theorem \ref{thm: mr}}\label{S:Main_Proof}

Before proving Theorem \ref{thm: mr} we state one equality and prove one technical fact.  For the equality, let $\hat{\phi}$ be from Lemma \ref{lem: min_val}. In light of \eqref{eq: diff_val} we obtain
\begin{equation}\label{eq: second_order_simp}
\begin{split}
\frac{L^c \hat{u}}{\hat{u}} &= \frac{1}{4}\tr\pare{cD^2\hat{\phi}} + \frac{1}{8}\nabla\hat{\phi}'c\nabla \hat{\phi};\\
&= \frac{1}{4p}\nabla\cdot(pc\ell) - \frac{1}{4}\nabla\hat{\phi}'c\ell + \frac{1}{8}\nabla\hat{\phi}'c\nabla\hat{\phi};\\
&= \frac{1}{4p}\nabla\cdot\left(pc\ell\right) + \frac{1}{8}\left(\nabla\hat{\phi}-\ell\right)'c\left(\nabla\hat{\phi}-\ell\right) - \frac{1}{8}\ell'c\ell.
\end{split}
\end{equation}

As for the technical fact, we have

\begin{lem}\label{lem: divergence_vanish} Let Assumptions \ref{ass: Ecp}--\ref{ass: Ecp_i}(iii) hold. Then $\int_E \nabla\cdot(pc\ell) = 0$.
\end{lem}

\begin{proof}[Proof of Lemma \ref{lem: divergence_vanish}]
If $\nabla\cdot(pc\ell) = 0$ for all $x\in E$ then clearly the result holds.  Else, let $\hat{\phi}$ be from Lemmas \ref{lem: min_val}, \ref{lem: ergodic} and note that $\hat{\phi}$ is not identically constant and hence $\int_E p\nabla\hat{\phi}'c\nabla\hat{\phi} > 0$. Recalling $X^R$ from Assumption \ref{ass: Ecp_i}(iii) and using \eqref{eq: second_order_simp}:
\begin{equation*}
\begin{split}
\frac{1}{T}\hat{\phi}(X^R_T) &= \frac{1}{T}\hat{\phi}(X^R_0) + \frac{1}{T}\int_0^T \nabla\hat{\phi}'c\ell (X^R_t)dt + \frac{1}{T}\int_0^T \nabla\hat{\phi}'\sigma(X^R_t)dW_t + \frac{1}{2}\int_0^T \tr\pare{cD^2\hat{\phi}}(X^R_t)dt;\\
&= \frac{1}{T}\hat{\phi}(X^R_0) + \frac{1}{T}\int_0^T \frac{1}{p}\nabla\cdot (pc\ell)(X^R_t)dt + \frac{1}{T}\int_0^T \nabla\hat{\phi}'\sigma(X^R_t)dW_t;\\
&= \frac{1}{T}\hat{\phi}(X^R_0) + \frac{1}{T}\int_0^T \frac{1}{p}\nabla\cdot (pc\ell)(X^R_t)dt + \frac{1}{T}\int_0^T \nabla\hat{\phi}'c\nabla\hat{\phi}(X^R_t)dt \left(\frac{\int_0^T \nabla\hat{\phi}'\sigma(X^R_t)dW_t}{\int_0^T \nabla\hat{\phi}'c\nabla\hat{\phi}(X^R_t)dt}\right).\\
\end{split}
\end{equation*}
Since $(\nabla\cdot(pc\ell))^+\in L^1(E,\textrm{leb})$, $\int_E p\nabla \hat{\phi}'c\nabla\hat{\phi}>0$, the Dambis-Dubins-Schwarz theorem and strong law for Brownian motion imply almost surely:
\begin{equation*}
\lim_{T\uparrow\infty}\frac{1}{T}\hat{\phi}(X^R_T) = \int_E \nabla\cdot(pc\ell).
\end{equation*}
If the value on the right hand side above were not zero, it would contradict the positive recurrence of $X^R$.
\end{proof}

\begin{proof}[Proof of Theorem \ref{thm: mr}]
From \eqref{eq: lambda_lb} we see that
\begin{equation}\label{eq: temp_ub_value}
\lambda \geq I = -\inf_{u\in\D}\int_E \frac{L^c u}{u}p \geq -\int_E \frac{L^c\hat{u}}{\hat{u}}p.
\end{equation}
For now, assume $\hat{u}= e^{(1/2)\hat{\phi}} \in\D$. By Lemma \ref{lem: ergodic}, the diffusion $\hat{X}$ from \eqref{eq: p_star_diffusion} is ergodic, and hence the associated $\hat{\prob}\in \Pi$. As $V^{\hat{\vartheta}}$ enjoys the num\'eraire property under $\hat{\prob}$, we know
\begin{equation*}
\lambda \leq -\int_E \frac{L^c\hat{u}}{\hat{u}}p \leq -\inf_{u\in D} \int_E \frac{L^c\hat{u}}{\hat{u}}p = I,
\end{equation*}
which in conjunction with \eqref{eq: temp_ub_value} establishes the first equality in \eqref{eq: lambda_I_equiv}, provided that $\hat{u}\in\D$. To show this latter fact, recall \eqref{eq: second_order_simp}.  Since $(\nabla\cdot (pc\ell))^+ \in L^1(E,\textrm{leb})$,  Lemma \ref{lem: min_val} implies $(L^c \hat{u}/u)^+ \in L^1(E,p)$, and hence $\hat{u}\in \D$.

It remains to prove the second equality in \eqref{eq: lambda_I_equiv} as well as that $\probgrowth{V^{\hat{\vartheta}}}{\prob} = \lambda$ for all $\prob\in\Pi$.  From \eqref{eq: second_order_simp} and Lemma \ref{lem: divergence_vanish} we see that
\begin{equation}\label{eq: temp_value_2}
\int_E \frac{L^c u}{u}p = \frac{1}{8}\int_E \left(\nabla\hat{\phi}-\ell\right)'c\left(\nabla\hat{\phi}-\ell\right)p - \frac{1}{8}\int_E \ell'c\ell p.
\end{equation}
Next, if $\nabla\cdot(pc\ell) = 0$ for $x\in E$ then $\hat{\phi}$ is constant and clearly, \eqref{eq: temp_value_2} implies the second equality in \eqref{eq: lambda_I_equiv}.  Else, $\nabla\hat{\phi}$ is not identically $0$, and $\int_E \nabla\hat{\phi}'c\nabla\hat{\phi}p >0$.  Continuing, note that for $\hat{X}$ as in \eqref{eq: diffusion_hat_phi} we have, using \eqref{eq: second_order_simp}
\begin{align*}
\hat{\phi}(\hat{X}_T) &= \hat{\phi}(\hat{X}_0) + \frac{1}{2}\int_0^T \nabla\hat{\phi}'c\nabla\hat{\phi}(\hat{X}_t)dt + \int_0^T \nabla\hat{\phi}'\sigma(\hat{X}_u)dW_u + \frac{1}{2}\int_0^T \tr\pare{cD^2\hat{\phi}}(\hat{X}_t)dt;\\
&= \hat{\phi}(\hat{X}_0) + \frac{1}{2}\int_0^T \nabla\hat{\phi}'c\nabla\hat{\phi}(\hat{X}_t)dt + \int_0^T \nabla\hat{\phi}'\sigma(\hat{X}_u)dW_u + \frac{1}{2}\int_0^T \left(\frac{1}{p}\nabla\cdot(pc\ell) - \nabla\hat{\phi}'c\ell\right)(\hat{X}_t)dt;\\
&= \hat{\phi}(\hat{X}_0) + \int_0^T\left(\frac{1}{4}\nabla\hat{\phi}'c\nabla\hat{\phi} - \frac{1}{4}\ell'c\ell + \frac{1}{4}\left(\nabla\hat{\phi}-\ell\right)'c\left(\nabla\hat{\phi} - \ell\right) + \frac{1}{2p}\nabla\cdot(pc\ell)\right)(\hat{X}_u)du\\
 &\qquad\qquad + \int_0^T \nabla\hat{\phi}'\sigma (\hat{X}_t)dW_t.
\end{align*}
So, we see by the Strong Law for Brownian Motion, the Dambis-Dubins-Schwarz theorem and the given assumptions, we have almost surely
\begin{align*}
\lim_{T\uparrow\infty} \frac{1}{T}\hat{\phi}(\hat{X}_T) &= \frac{1}{4}\int_E \left(\nabla\hat{\phi}'c\nabla\hat{\phi} - \ell'c\ell + \left(\nabla\hat{\phi}-\ell\right)'c\left(\nabla\hat{\phi} - \ell\right)\right)p + \frac{1}{2}\int_E \nabla\cdot(pc\ell);\\
&= \frac{1}{4}\int_E \left(\nabla\hat{\phi}'c\nabla\hat{\phi} - \ell'c\ell + \left(\nabla\hat{\phi}-\ell\right)'c\left(\nabla\hat{\phi} - \ell\right)\right)p,
\end{align*}
where the last inequality follows by Lemma \ref{lem: divergence_vanish}. Now, if the right hand side above was not zero it would violate the positive recurrence of $\hat{X}$. This gives
\begin{equation}\label{eq: two_integral_the_same}
\begin{split}
\int_E \nabla\hat{\phi}'c\nabla\hat{\phi}p = \int_E \ell'c\ell p - \int_E \left(\nabla\hat{\phi}-\ell\right)'c\left(\nabla\hat{\phi} - \ell\right)p,
\end{split}
\end{equation}
which, in view of \eqref{eq: temp_value_2}, establishes the second equality in \eqref{eq: lambda_I_equiv}.

The last thing to show is $G(V^{\hat{\vartheta}},\prob) = \lambda$ for all $\prob\in\Pi$.  To this end, by \ito's formula and \eqref{eq: second_order_simp} we know
\begin{equation}\label{E:probrowth_P_abs}
\begin{split}
\frac{1}{T}\log V^{\hat{\phi}}_T &= \frac{1}{T}\log V_0 + \frac{1}{2T}\hat{\phi}(X_T)-\frac{1}{2T}\hat{\phi}(X_0) - \frac{1}{8T}\int_0^T\left(2\tr\pare{CD^2\hat{\phi}}+\nabla\hat{\phi}'c\nabla\hat{\phi}\right)(X_t)dt;\\
&= \frac{1}{T}\log V_0 + \frac{1}{2T}\hat{\phi}(X_T)-\frac{1}{2T}\hat{\phi}(X_0) - \frac{1}{4T}\int_0^T \frac{1}{p}\nabla\cdot\left(pc\ell\right)(X_t)dt\\
&\qquad\qquad + \frac{1}{8T}\int_0^T \ell'c\ell(X_t)dt - \frac{1}{8T}\int_0^T \left(\nabla\hat{\phi}-\ell\right)'c\left(\nabla\hat{\phi}-\ell\right)(X_t)dt.
\end{split}
\end{equation}
Taking $T\uparrow\infty$ gives
\begin{equation*}
\begin{split}
\probgrowth{V^{\hat{\phi}}}{\prob} &= -\frac{1}{4}\int_E \nabla\cdot(pc\ell) + \frac{1}{8}\int_E \ell'c\ell p - \frac{1}{8}\int_E \left(\nabla\hat{\phi}-\ell\right)'c\left(\nabla\hat{\phi}-\ell\right)p;\\
&= \frac{1}{8}\int_E \ell'c\ell p - \frac{1}{8}\int_E \left(\nabla\hat{\phi}-\ell\right)'c\left(\nabla\hat{\phi}-\ell\right)p;\\
&=\lambda,
\end{split}
\end{equation*}
where the second equality came from Lemma \ref{lem: divergence_vanish}, and the third from \eqref{eq: lambda_I_equiv}, \eqref{eq: two_integral_the_same}.  This finishes the proof.
\end{proof}

\section{Proofs from Section \ref{S:RBA}}

\label{S:proofs_from_rank}

We keep all notation from Section \ref{S:RBA}. Additionally, we denote $\T$ as the set of all permutations $\tau$ of $\cbra{1, \ldots ,d}$. For $\tau \in \T$, and $x \in \splex$, we define $x_\tau  \in \splex$ by $x_\tau^i = x(\tau^i)$, $i=1, \ldots ,d$.

\subsection{Proof of Proposition \ref{prop: growth_same} }

In the course of the proof, we shall use the sets
\[
R_\tau \dfn \cbra{x\in\splex\such x_\tau \in \osplex}, \quad \tau \in \T.
\]
Note that the $\cbra{R_\tau \such \tau\in\T}$ may not be disjoint, but their topological interiors are.

We first show that $\hat{u}$ from \eqref{eq:minimization} is permutation invariant.  To this end, recall that $\hat{u} = \exp (\hat{\phi}/2)$, where $\hat{\phi}$ solves the variational problem in Lemma \ref{lem: min_val}, and recall the functional $J(\phi)$ from \eqref{eq: ell_J_def}.  For a given $\tau\in\T$ and function $\phi$, write $\phi_\tau(x) \dfn \phi(x_\tau)$, $x \in \splex$.  We claim that
\begin{equation}\label{eq: J_permut}
J(\phi) = J(\phi_\tau);\qquad \forall \tau\in\T.
\end{equation}
Admitting \eqref{eq: J_permut}, that $\hat{\phi}(x) = \hat{\phi}(x_\tau)$ (and hence $\hat{u}(x)=\hat{u}(x_\tau)$) for all $\tau\in \T$ is easy to show. Indeed, as the functional $J(\phi)$ is evidently convex, we see that
\begin{equation*}
J\left(\frac{1}{d!}\sum_{\tau} \phi_\tau\right) \leq \frac{1}{d!}\sum_\tau J(\phi_\tau) = J(\phi),
\end{equation*}
where the last equality follows by \eqref{eq: J_permut}.  Thus, if $\hat{\phi}$ is a minimizer then so is $(1/d!)\sum_{\tau} \hat{\phi}_\tau$ and by Lemma \ref{lem: min_val} we can write
\begin{equation*}
\hat{\phi} = \frac{1}{d!}\sum_{\tau} \hat{\phi}_\tau + c,
\end{equation*}
for some constant $c$. But, as the right hand side above is permutation invariant, so is the left hand side. It remains to prove \eqref{eq: J_permut}, which will follow by straight-forward  computations, and which uses the following identities for $\tau\in\T$:
\begin{equation}\label{eq: permut_ident}
\begin{split}
f(x) &= g(x_\tau) \Longrightarrow \partial_i f(x) = \partial_{\tauinv(i)} g(x_\tau);\\
p(x) &= p(x_\tau),
\end{split}
\end{equation}
and
\begin{equation}\label{eq: permut_ident_A}
\begin{split}
c^{ij}(x) &= c^{\tauinv(i)\tauinv(j)}(x_\tau);\\
\partial_j c^{ij}(x) &= \partial_{\tauinv(j)} c^{\tauinv(i)\tauinv(j)}(x_\tau).
\end{split}
\end{equation}
Showing \eqref{eq: permut_ident} is straight-forward. As for the first equality in \eqref{eq: permut_ident_A}, we have
\begin{equation*}
\begin{split}
c^{\tauinv(i)\tauinv(j)}(x_\tau) &= \kappa^{r(x(\tau(\tauinv(i)))) r(x(\tau(\tauinv(j))))}(\xord);\\
&=\kappa^{r(x^i)r(x^j)}(\xord);\\
&=c^{ij}(x).
\end{split}
\end{equation*}
The second equality in \eqref{eq: permut_ident_A} follows from the first as well as \eqref{eq: permut_ident}. Now, plugging in for $\ell$ from \eqref{eq: ell_def} we have
\begin{equation*}
\begin{split}
J(\phi)- \int_{\splex} p\ell'c\ell &= \int_{\splex} p\nabla \phi' c\nabla \phi - 2\int_{\splex} \nabla p' c \nabla \phi - 2\int_{\splex} p \nabla \phi'\matdiv{c};\\
&\dfn \mathbf{A}(\phi) + \mathbf{B}(\phi) + \mathbf{C}(\phi).
\end{split}
\end{equation*}
We handle the three terms separately and repeatedly use \eqref{eq: permut_ident}, \eqref{eq: permut_ident_A}. Also, we will omit the summands. As for $\mathbf{A}$, assume $x\in R_\tau$ so that $\xord = x_\tau$. Then
\begin{equation}\label{eq: bold_A}
\begin{split}
\int_{\splex} p(x)\partial_i(\phi_\tau)(x)c^{ij}(x)\partial_j(\phi_\tau)(x)dx &= \int_{\splex}p(x_\tau) \partial_{\tauinv(i)}\phi(x_\tau)c^{\tauinv(i)\tauinv(j)}(x_\tau)\partial_{\tauinv(j)}\phi(x_\tau)dx; \\
&= \int_{\splex} p(y)\partial_{a}\phi(y)c^{ab}(y)\partial_{b}\phi(y)dy,
\end{split}
\end{equation}
where to get the last equality we let $y=x_\tau$ and noted that $dy=dx$; and set $a=\tauinv(i),b=\tauinv(j)$.  This shows $\mathbf{A}(\phi_\tau) = \mathbf{A}(\phi)$.  As for $\mathbf{B}$:
\begin{equation}\label{eq: bold_B}
\begin{split}
\int_{\splex} \partial_i(p)(x)c^{ij}(x)\partial_j(\phi_\tau)(x)dx&= \int_{\splex} \partial_{\tauinv(i)}p(x_\tau)c^{\tauinv(i)\tauinv(j)}(x_\tau)\partial_{\tauinv(j)}\phi(x_\tau)dx; \\
&= \int_{\splex} \partial_a p(y)c^{ab}(y)\partial_{b}\phi(y)dy.
\end{split}
\end{equation}
Thus, $\mathbf{B}(\phi_\tau)=\mathbf{B}(\phi)$. Lastly, for $\mathbf{C}$:
\begin{equation}\label{eq: bold_C}
\begin{split}
\int_{\splex} p(x)\partial_i(\phi_\tau)(x)\partial_j c^{ij}(x)dx &= \int_{\splex} p(x_\tau)\partial_{\tauinv(i)}\phi(x_\tau)\partial_{\tauinv(j)}c^{\tauinv(i)\tauinv(j)}(x_\tau)dx; \\
&= \int_{\splex} p(y)\partial_{a}\phi(y)\partial_b c^{ab}(y)dy.
\end{split}
\end{equation}
Thus, $\mathbf{C}(\phi_\tau)=\mathbf{C}(\phi)$ and hence \eqref{eq: J_permut} holds.

The third (last) equality in \eqref{eq: functionals} holds in view of Theorem \ref{thm: mr}. Next, we show the second equality in \eqref{eq: functionals}. To do so, we will prove three equalities, analogous to \eqref{eq: bold_A}, \eqref{eq: bold_B} and \eqref{eq: bold_C}, which all follow by construction of $p,c$, since $\phi(x)=\phi(x_\tau)$, and \eqref{eq: permut_ident}, \eqref{eq: permut_ident_A}.  Proceeding, let $\tau\in\T$ and $x\in R_\tau$. We first have
\begin{equation}\label{eq: bold_AA}
\begin{split}
\partial_i\phi(x)c^{ij}(x)\partial_j\phi(x) &= \partial_{\tauinv(i)}\phi(x_\tau)c^{\tauinv(i)\tauinv(j)}(x_\tau)\partial_{\tauinv(j)}\phi(x_\tau);\\
&=  \partial_{a}\phi(x_\tau)\kappa^{ab}(x_\tau)\partial_{b}\phi(x_\tau);\\
&=  \partial_{a}\phi(\xord)\kappa^{ab}(\xord)\partial_{b}\phi(\xord).
\end{split}
\end{equation}
Next, we have
\begin{equation}\label{eq: bold_BB}
\begin{split}
&\partial_i \phi(x)\left(c^{ij}(x)\frac{\partial_j p(x)}{p(x)} + \partial_j c^{ij}(x)\right);\\
&\qquad = \partial_{\tauinv(i)}\phi(x_\tau)\left(c^{\tauinv(i)\tauinv(j)}(x_\tau)\frac{\partial_{\tauinv(j)}p(x_\tau)}{p(x_\tau)} + \partial_{\tauinv(j)}c^{\tauinv(i)\tauinv(j)}(x_\tau)\right);\\
&\qquad = \partial_{a}\phi(x_\tau)\left(c^{ab}(x_\tau)\frac{\partial_{b}p(x_\tau)}{p(x_\tau)} + \partial_{b}c^{ab}(x_\tau)\right);\\
&\qquad = \partial_{a}\phi(\xord)\left(\kappa^{ab}(\xord)\frac{\partial_{b}q(\xord)}{q(\xord)} + \partial_{b}\kappa^{ab}(\xord)\right),\\
\end{split}
\end{equation}
where the last equality holds because $c(x_\tau)=\kappa(x_\tau)=\kappa(\xord)$ and $p(x_\tau) = (1/d!)q(x_\tau) = (1/d!)q(\xord)$ in $R_\tau$. Lastly, define
\begin{equation}\label{eq: ranked_ell}
\ell_{\leq}(\xord)\dfn \left(\frac{\nabla q}{q} + \kappa^{-1}\matdiv{\kappa}\right)(\xord);\qquad \xord\in\osplex.
\end{equation}
We then have
\begin{equation}\label{eq: bold_CC}
\begin{split}
\frac{1}{p(x)}\partial_i\left(pc\ell\right)^i(x) & = c^{ij}(x)\frac{\partial_{ij}p(x)}{p(x)} + 2\frac{\partial_i p(x)}{p(x)}\partial_j c^{ij}_j(x) + \partial_{ij}c^{ij}(x);\\
&= c^{\tauinv(i)\tauinv(j)}(x_\tau)\frac{\partial_{\tauinv(i)\tauinv(j)}p(x_\tau)}{p(x_\tau)} + 2\frac{\partial_{\tauinv(i)}p(x_\tau)}{p(x_\tau)}\partial_{\tauinv(j)}c^{\tauinv(i)\tauinv(j)}(x_\tau)\\
&\qquad  + \partial_{\tauinv(i)\tauinv(j)}c^{\tauinv(i)\tauinv(j)}(x_\tau);\\
&= c^{ab}(x_\tau)\frac{\partial_{ab}p(x_\tau)}{p(x_\tau)} + 2\frac{\partial_{a}p(x_\tau)}{p(x_\tau)}\partial_{b}c^{ab}(x_\tau) + \partial_{ab}c^{ab}(x_\tau);\\
&= \kappa^{ab}(\xord)\frac{\partial_{ab}q(\xord)}{p(\xord)} + 2\frac{\partial_{a}q(\xord)}{q(\xord)}\partial_{b}\kappa^{ab}(\xord) + \partial_{ab}\kappa^{ab}(\xord);\\
&= \frac{1}{q(\xord)}\partial_a\left(q \kappa\ell_{\leq}\right)^a(\xord).
\end{split}
\end{equation}
Since \eqref{eq: bold_AA}, \eqref{eq: bold_BB}, \eqref{eq: bold_CC} hold for $x\in R_\tau$ for any $\tau\in\mathcal{T}$, they in fact hold for all $x\in\splex$. Thus, from \eqref{eq: bold_AA} we obtain
\begin{equation*}
\begin{split}
\int_{\splex} \left(p\nabla\phi'c\nabla\phi\right)(x) &= \frac{1}{d!}\int_{\splex} \left(q\nabla\phi'\kappa\nabla \phi\right)(\xord) = \int_{\osplex}\left(q\nabla\phi'\kappa\nabla \phi\right)(\xord),
\end{split}
\end{equation*}
which is the second equality in \eqref{eq: functionals}.

Continuing, we show \eqref{eq: ranked_probgrowth}.  Let $\prob\in\Pi_{\leq}$. \ito's formula, \eqref{eq:quad_var} and \eqref{eq: c_ranked} give
\begin{equation*}
\begin{split}
\frac{1}{T}\log\left(V^{\hat{\vartheta}}_T\right) &= \frac{1}{T}\log\left(\frac{\hat{u}(X_T)}{\hat{u}(X_0)}\right) - \frac{1}{2T}\int_0^T \frac{1}{\hat{u}(X_t)}\sum_{i,j=1}^d \partial^2_{ij}\hat{u}(X_t)\kappa^{r(X^i_t)r(X^j_t)}\left(\bigxord_t\right)dt;\\
&=\frac{1}{T}\log\left(\frac{\hat{u}(X_T)}{\hat{u}(X_0)}\right) - \frac{1}{T}\int_0^T\frac{L^c\hat{u}}{\hat{u}}(X_t)dt.
\end{split}
\end{equation*}
From \eqref{eq: second_order_simp} and $\hat{u}=e^{(1/2)\hat{\phi}}$ we obtain
\begin{equation*}
\begin{split}
\frac{1}{T}\log\left(V^{\hat{\vartheta}}_T\right) &= \frac{1}{2T}\hat{\phi}(X_T) - \frac{1}{2T}\hat{\phi}(X_0) - \frac{1}{4T}\int_0^T\frac{1}{p}\nabla\cdot\left(pc\ell\right)(X_t)dt - \frac{1}{8T}\int_0^T \nabla\hat{\phi}'c\nabla\hat{\phi}(X_t)dt\\
&\qquad + \frac{1}{4T}\int_0^T \nabla\hat{\phi}'\left(c\frac{\nabla p}{p} + \matdiv{c}\right)(X_t)dt;\\
&= \frac{1}{2T}\hat{\phi}(X_T) - \frac{1}{2T}\hat{\phi}(X_0) - \frac{1}{4T}\int_0^T\frac{1}{q}\nabla\cdot\left(q\kappa\ell_{\leq}\right)(\bigxord_t)dt - \frac{1}{8T}\int_0^T \nabla\hat{\phi}'\kappa\nabla\hat{\phi}(\bigxord_t)dt\\
&\qquad + \frac{1}{4T}\int_0^T \nabla\hat{\phi}'\left(\kappa\frac{\nabla q}{q} + \matdiv{\kappa}\right)(\bigxord_t)dt;\\
&= \frac{1}{2T}\hat{\phi}(X_T) - \frac{1}{2T}\hat{\phi}(X_0) - \frac{1}{4T}\int_0^T\frac{1}{q}\nabla\cdot\left(q\kappa\ell_{\leq}\right)(\bigxord_t)dt\\
&\qquad - \frac{1}{8T}\int_0^T \left(\nabla\hat{\phi}-\ell_{\leq}\right)'\kappa\left(\nabla\hat{\phi}-\ell_{\leq}\right)(\bigxord_t)dt + \frac{1}{8T}\int_0^T \ell_{\leq}'\kappa\ell_{\leq}(\bigxord_t)dt,
\end{split}
\end{equation*}
where the second to last equality follows from \eqref{eq: bold_A}, \eqref{eq: bold_B}, \eqref{eq: bold_C}. These equalities, in conjunction with the integrability assumptions of Assumption \ref{ass: Ecp_i}, and $\prob\in\Pi_{\leq}$ allow us to deduce
\begin{equation*}
\begin{split}
\probgrowth{V^{\hat{\vartheta}}}{\prob} &= \frac{1}{4}\int_{\osplex}\nabla\cdot\left(q\kappa\ell_{\leq}\right)(\xord) - \frac{1}{8}\int_{\osplex} \left(\left(\nabla\hat{\phi}-\ell_{\leq}\right)'\kappa\left(\nabla\hat{\phi}-\ell_{\leq}\right)q\right)(\xord)\\
&\qquad  + \frac{1}{8}\int_{\osplex}\left(\ell_{\leq}'\kappa\ell_{\leq}q\right)(\xord);\\
&= \frac{1}{4}\int_{\osplex}\nabla\cdot\left(q\kappa\ell_{\leq}\right)(\xord) - \frac{1}{8}\int_{\osplex} \left(\nabla\hat{\phi}'\kappa\nabla\hat{\phi}q\right)(\xord)\\
&\qquad  + \frac{1}{4}\int_{\osplex}\left(\nabla\hat{\phi}'\left(\kappa\nabla q + q\matdiv{\kappa}\right)\right)(\xord);\\
&= \frac{1}{4}\int_{\splex}\nabla\cdot\left(p c\ell\right)(x) - \frac{1}{8}\int_{\splex} \left(\nabla\hat{\phi}'c\nabla\hat{\phi}p\right)(x)\\
&\qquad  + \frac{1}{4}\int_{\splex}\left(\nabla\hat{\phi}'\left(c\nabla p + p\matdiv{c}\right)\right)(x);\\
&= \frac{1}{4}\int_{\splex}\nabla\cdot\left(p c\ell\right)(x) - \frac{1}{8}\int_{\splex} \left(\left(\nabla\hat{\phi}-\ell\right)'c\left(\nabla\hat{\phi}-\ell\right)p\right)(x)\\
&\qquad  + \frac{1}{8}\int_{\splex}\left(\ell'c\ell p\right)(x);\\
&= \frac{1}{8}\int_{\splex}\left(\nabla\hat{\phi}'c\nabla\hat{\phi}p\right)(x).
\end{split}
\end{equation*}
Above, the third equality holds again because of \eqref{eq: bold_AA}, \eqref{eq: bold_BB}, \eqref{eq: bold_CC}.  The fifth inequality follows because of Lemma \ref{lem: divergence_vanish} and \eqref{eq: two_integral_the_same} in the proof of Theorem \ref{thm: mr}.  This, and the fact we have already proved the second and third equalities in \eqref{eq: functionals} , yields \eqref{eq: ranked_probgrowth} since $\hat{u} = e^{(1/2)\hat{\phi}}$.

It remains to prove that $\lambda_{\leq} = \lambda$, which will establish all the equalities in \eqref{eq: functionals}.  Recall that $\lambda_{\leq} \leq \lambda$ holds from \eqref{eq: ranked_abs_ub}. Furthermore, using \eqref{eq: ranked_probgrowth} and the last equality in \eqref{eq: functionals},
\begin{equation*}
\lambda_{\leq} \geq \inf_{\prob\in\Pi_{\leq}} \probgrowth{V^{\hat{\vartheta}}}{\prob} = \frac{1}{2} \int_{\splex} \pare{\frac{\nabla \hat{u}}{\hat{u}}}' c \pare{\frac{\nabla \hat{u}}{\hat{u}}}p =\lambda.
\end{equation*}
Thus, $\lambda = \lambda_{\leq}$ and the proof is finished.
\qed

\subsection{Proof of Proposition \ref{prop: modify}}

We start with the construction of a particular matrix valued function which works well with Assumption \ref{ass: kappa_q_extend}.  To state the following auxiliary result, define
\begin{equation}\label{E:x_min_max}
\ol{x}\dfn \max\cbra{x^1,\ldots,x^d}, \quad \ul{x} \dfn \min\cbra{x^1,\ldots,x^d}; \qquad x\in\splex.
\end{equation}

\begin{lem}\label{lem: nice_matrix}
Let $A,B,C\in\reals$ be such that 1) $C\geq 0$, 2) $B\leq A < 2B$ and 3) $A+C\geq 2$.  For $x\in \splex$ define the matrix $\theta$ via
\begin{equation}\label{eq: theta_def}
\theta^{ij}(x) \dfn 1_{i=j}\left((x^i)^A \prod_{l=1}^d (x^l)^C\right) + 1_{i\neq j}\left((x^i)^B (x^j)^B\prod_{l=1}^d (x^l)^{A+C-B}\right);\qquad i,j = 1,\ldots,d.
\end{equation}
Then:
\begin{enumerate}
	\item For any $\tau\in\T$, we have $\theta^{\tauinv(i)\tauinv(j)}(x_\tau) = \theta^{ij}(x)$.
	\item For every $\xi\in\reals^d$ we have
		\begin{equation*}
		\xi'\theta(x)\xi \geq k(x)\xi'\xi;\qquad k(x)\dfn \left(\prod_{l=1}^d (x^l)^C\right)\left(1-\ol{x}^{2B-A}\right)\min\cbra{1,\ul{x}^A}.
		\end{equation*}
	\item $\theta$ is smooth in $\splex$ and the diffusion
		\begin{equation*}
		dX_t = \frac{1}{2}\matdiv{\theta}(X_t)dt + \sqrt{\theta(X_t)} dW_t,
		\end{equation*}
		does not explode to $\partial \Delta_{+}$.
	\item $\int_{\Delta_+}\left|\nabla\cdot\left(\matdiv{\theta}\right)\right| < \infty$.
	\item $\int_{\Delta_+}\matdiv{\theta}'\theta^{-1}\matdiv{\theta} < \infty$.
	\item $\theta^{-1}\matdiv{\theta} = \nabla H$ where $H(x) = (A+C)\log\left(\prod_{l=1}^d x^l\right)$.
\end{enumerate}
\end{lem}

\begin{proof}
We tackle each point below.
	\begin{enumerate}
		\item We have
		\begin{align*}
		\theta^{\tauinv(i)\tauinv(j)}(x_\tau) &= 1_{i=j}\left(x_\tau(\tauinv(i))^A \prod_{l=1}^d x_\tau(l)^C\right) + 1_{i\neq j}\left(x_\tau(\tauinv(i))^B x_\tau(\tauinv(j))^B \prod_{l=1}^d x_\tau(l)^{A+C-B}\right);\\
		&= 1_{i=j}\left((x^i)^A \prod_{l=1}^d (x^l)^C\right) + 1_{i\neq j}\left((x^i)^B (x^j)^B \prod_{l=1}^d (x^l)^{A+C-B}\right);\\
		&= \theta^{ij}(x).
		\end{align*}
		\item We have
		\begin{align*}
		\xi'\theta(x)\xi &= \sum_{i=1}^d \xi(i)^2 (x^i)^A \prod_{l=1}^d (x^l)^C + \sum_{i,j=1, i\neq j}^d \xi(i)\xi(j)(x^i)^B(x^j)^B \prod_{l=1}^d (x^l)^{A+C-B};\\
		&= \sum_{i=1}^d \xi(i)^2\left((x^i)^A\prod_{l=1}^d (x^l)^C - (x^i)^{2B}\prod_{l=1}^d (x^l)^{A+C-B}\right) + \left(\prod_{l=1}^d (x^l)^{A+C-B}\right)\left(\sum_{i=1}^d \xi(i)(x^i)^B\right)^2;\\
		&\geq \sum_{i=1}^d \xi(i)^2\left((x^i)^A\prod_{l=1}^d (x^l)^C - (x^i)^{2B}\prod_{l=1}^d (x^l)^{A+C-B}\right);\\
		&= \left(\prod_{l=1}^d (x^l)^C\right)\sum_{i=1}^d \xi(i)^2 (x^i)^A\left(1 - (x^i)^{2B-A}\prod_{l=1}^d (x^l)^{A-B}\right).
		\end{align*}
		Since $A\geq B$ we have $\prod_{l=1}^d (x^l)^{A-B}\leq 1$.  Since $2B-A>0$ we have $(x^i)^{2B-A} \leq \ol{x}^{2B-A} < 1$.  Lastly, we have $(x^i)^A \geq \min\cbra{1, \ul{x}^A}$.  Putting these together gives the claim.
		\item We have
		\begin{equation*}
		\theta^{ij}_j = 1_{i=j}\left((A+C)(x^i)^{A+C-1}\prod_{l\neq i} (x^l)^C\right) + 1_{i\neq j}\left((x^i)^{A+C} (A+C)(x^j)^{A+C-1}\prod_{l\neq i,j} (x^l)^{A+C-B}\right).
		\end{equation*}
		Thus, we see that
		\begin{align*}
		\matdiv{\theta}^i &= \sum_{j} \theta^{ij}_j = (A+C)\left((x^i)^{A+C-1}\prod_{l\neq i}(x^l)^C + (x^i)^{A+C}\sum_{j\neq i} (x^j)^{A+C-1}\prod_{l\neq i,j} (x^l)^{A+C-B}\right);\\
		&= x^iY_i,
		\end{align*}
		where
		\begin{equation*}
		Y_i \dfn (A+C)(x^i)^{A+C-2}\left(\prod_{l\neq i}(x^l)^C + x^i\sum_{j\neq i}(x^j)^{A+C-1}\prod_{l\neq i,j}(x^l)^{A+C-B}\right).
		\end{equation*}
		Since $C\geq 0$ and $A+C\geq 2$ we see that
		\begin{equation*}
		0 \leq Y_i \leq d(A+C).
		\end{equation*}
		In a similar manner we have
		\begin{equation*}
		\theta^{ii}(x) = (x^i)^2 Z_i^2;\qquad Z_i\dfn (x^i)^{(A+C-2)/2}\prod_{l\neq i} (x^l)^{C/2}.
		\end{equation*}
		Again, the given hypotheses yield that $0 \leq Z_i \leq  1$.  Now, let $X(t)$ be a local solution (i.e. up to first exit time $\tau$ of some set compactly contained within $\splex$) to the above SDE.  We have that for $t\leq \tau$ that
		\begin{equation*}
		dX_t(i) = \frac{1}{2}X_t(i)Y_i(X_t)dt + X_t(i)Z_i(X_t)d\tilde{B}_t,
		\end{equation*}
		for a Brownian motion $\tilde{B}$.  Now, since $Y_i$ and $Z_i$ are bounded on $\Delta_{+}$ it is clear that $X(i)$ does not hit zero for any $i$.  This gives the result.
		\item We have from $(3)$ above that
		\begin{align*}
		&\partial_i\left(\sum_{j} \theta^{ij}_{j}\right)\\
		\quad &= (A+C)\left((A+C-1)(x^i)^{A+C-2}\prod_{l\neq i} (x^l)^C + (A+C)(x^i)^{A+C-1}\sum_{j\neq i} (x^j)^{A+C-1}\prod_{l\neq i,j} (x^l)^{A+C-B}\right).
		\end{align*}
		Since $A+C\geq 2, C\geq 0$ and $A\geq B$ we see that
		\begin{equation*}
		\left|\nabla\cdot(\matdiv{\theta})\right| = \left|\sum_{i,j}\theta^{ij}_{ij}\right| \leq (A+C)\left(d(A+C-1) + (d-1)(A+C)\right),
		\end{equation*}
		from which the result follows.
		\item Write $\tau = \theta^{-1}\matdiv{\theta}$ so that $\matdiv{\theta} = \theta\tau$.  Plugging in for $\theta,\matdiv{\theta}$ we see that
		\begin{align*}
		\matdiv{\theta}^i &= (A+C)\left((x^i)^{A+C-1}\prod_{l\neq i}(x^l)^C + (x^i)^{A+C}\sum_{j\neq i}(x^j)^{A+C-1}\prod_{l\neq i,j}(x^l)^{A+C-B}\right);\\
		(\theta\tau)^i &= (x^i)^{A+C}\prod_{l\neq i}(x^l)^C \ell^i + (x^i)^{A+C}\sum_{j\neq i}(x^j)^{A+C}\prod_{l\neq i,j}(x^l)^{A+C-B} \ell^j.
		\end{align*}
		From here, it is clear that $\tau^i = (A+C)/x^i$. Therefore, we have
		\begin{equation*}
		\begin{split}
		\matdiv{\theta}'\theta^{-1}\matdiv{\theta} &= \tau'c\tau\\
		&= (A+C)^2\left(\sum_i (x^i)^{A+C-2}\prod_{l\neq i}(x^l)^C + \sum_{i\neq j} (x^i)^{A+C-1}(x^j)^{A+C-1}\prod_{l\neq i,j}(x^l)^{A+C-B}\right);\\
		&\leq d^2(A+C)^2,
		\end{split}
		\end{equation*}
		and hence the result holds.
		\item We just showed that $(\theta^{-1}\matdiv{\theta})^i = \tau^i  = (A+C)/x^i$ for $i=1,..,d$.  Thus, the result follows since $\nabla\left(\prod_l x^l\right)^i = 1/x^i$.
	\end{enumerate}
\end{proof}

We are now in position to give the proof of Proposition \ref{prop: modify}.

\begin{proof}[Proof of Proposition \ref{prop: modify}]
	
	Assume that $V$ is an open subset of $\splex$ such that $\bar{V}\subset W$ with $W$ open and $\bar{W} \subseteq \osplex$.  As such $\textrm{dist}(V,\partial\osplex) > \delta > 0$ and we may find a  $C^\infty$ function $\chi$ on  $\osplex$ with $0\leq \chi\leq 1$, $\chi = 1$ on $V$ and $\chi(x) = 0$ if $\textrm{dist}(x,\partial\osplex) \leq \delta/3$, for example.  For $\theta$ as in Lemma \ref{lem: nice_matrix} we then set
	\begin{equation}\label{eq: qv_kappav}
	\begin{split}
	\kappa_V(x) &\dfn \chi(x)\kappa(x) + (1-\chi(x))\theta(x);\\
	q_V(x) &\dfn \chi(x)q(x) + (1-\chi(x))\frac{1-\int_{\osplex}\chi q}{\int_{\osplex}(1-\chi)}.
	\end{split}
	\end{equation}
	Now, create $c_V,p_V$ as in \eqref{eq: c_ranked}, \eqref{eq: p_ranked} respectively. By construction of $\theta$ in Lemma \ref{lem: nice_matrix}, we see that for any $x\in\splex$ such that $\chi(\xord) = 1$, with the $\tau$ such that $\xord=x_\tau$:
	\begin{equation*}
	c_V^{ij}(x) = \theta^{r(x^i)r(x^j)}(\xord) = \theta^{\tauinv(i)\tauinv(j)}(x_\tau) = \theta^{ij}(x),
	\end{equation*}
	where the last equality follows from Lemma \ref{lem: nice_matrix}. Thus, we see that $c$ is smooth in $\splex$.  The rest of the conditions in Assumptions \ref{ass: Ecp}, \ref{ass: Ecp_i} readily follow from Lemma \ref{lem: nice_matrix} as $q_v$ is constant near the boundary of $\partial \osplex$.
\end{proof}

\bibliographystyle{siam}
\bibliography{master}

\def\polhk#1{\setbox0=\hbox{#1}{\ooalign{\hidewidth
  \lower1.5ex\hbox{`}\hidewidth\crcr\unhbox0}}}
\begin{thebibliography}{10}

\bibitem{MR929084}
{\sc P.~H. Algoet and T.~M. Cover}, {\em Asymptotic optimality and asymptotic
  equipartition properties of log-optimum investment}, Ann. Probab., 16 (1988),
  pp.~876--898.

\bibitem{MR2192832}
{\sc H.~Attouch, G.~Buttazzo, and G.~Michaille}, {\em Variational analysis in
  {S}obolev and {BV} spaces}, vol.~6 of MPS/SIAM Series on Optimization,
  Society for Industrial and Applied Mathematics (SIAM), Philadelphia, PA;
  Mathematical Programming Society (MPS), Philadelphia, PA, 2006.
\newblock Applications to PDEs and optimization.

\bibitem{MR2428716}
{\sc A.~D. Banner and R.~Ghomrasni}, {\em Local times of ranked continuous
  semimartingales}, Stochastic Process. Appl., 118 (2008), pp.~1244--1253.

\bibitem{MR3114918}
{\sc E.~Bayraktar and Y.-J. Huang}, {\em Robust maximization of asymptotic
  growth under covariance uncertainty}, Ann. Appl. Probab., 23 (2013),
  pp.~1817--1840.

\bibitem{cuchiero2016cover}
{\sc C.~Cuchiero, W.~Schachermayer, and T.-K.~L. Wong}, {\em Cover's universal
  portfolio, stochastic portfolio theory and the numeraire portfolio}, arXiv
  preprint arXiv:1611.09631,  (2016).

\bibitem{MR997938}
{\sc J.-D. Deuschel and D.~W. Stroock}, {\em Large deviations}, vol.~137 of
  Pure and Applied Mathematics, Academic Press Inc., Boston, MA, 1989.

\bibitem{MR0386024}
{\sc M.~D. Donsker and S.~R.~S. Varadhan}, {\em Asymptotic evaluation of
  certain {M}arkov process expectations for large time. {I}. {II}}, Comm. Pure
  Appl. Math., 28 (1975), pp.~1--47; ibid. {\bf 28} (1975), 279--301.

\bibitem{MR0428471}
\leavevmode\vrule height 2pt depth -1.6pt width 23pt, {\em Asymptotic
  evaluation of certain {M}arkov process expectations for large time. {III}},
  Comm. Pure Appl. Math., 29 (1976), pp.~389--461.

\bibitem{MR690656}
\leavevmode\vrule height 2pt depth -1.6pt width 23pt, {\em Asymptotic
  evaluation of certain {M}arkov process expectations for large time. {IV}},
  Comm. Pure Appl. Math., 36 (1983), pp.~183--212.

\bibitem{MR1625845}
{\sc L.~C. Evans}, {\em Partial differential equations}, vol.~19 of Graduate
  Studies in Mathematics, American Mathematical Society, Providence, RI, 1998.

\bibitem{MR2676936}
{\sc D.~Fernholz and I.~Karatzas}, {\em On optimal arbitrage}, Ann. Appl.
  Probab., 20 (2010), pp.~1179--1204.

\bibitem{MR2732837}
\leavevmode\vrule height 2pt depth -1.6pt width 23pt, {\em Probabilistic
  aspects of arbitrage}, in Contemporary quantitative finance, Springer,
  Berlin, 2010, pp.~1--17.

\bibitem{MR1894767}
{\sc E.~R. Fernholz}, {\em Stochastic portfolio theory}, vol.~48 of
  Applications of Mathematics (New York), Springer-Verlag, New York, 2002.
\newblock Stochastic Modelling and Applied Probability.

\bibitem{MR1861997}
{\sc R.~Fernholz}, {\em Equity portfolios generated by functions of ranked
  market weights}, Finance Stoch., 5 (2001), pp.~469--486.

\bibitem{MR1814364}
{\sc D.~Gilbarg and N.~S. Trudinger}, {\em Elliptic partial differential
  equations of second order}, Classics in Mathematics, Springer-Verlag, Berlin,
  2001.
\newblock Reprint of the 1998 edition.

\bibitem{MR2932547}
{\sc P.~Guasoni and S.~Robertson}, {\em Portfolios and risk premia for the long
  run}, Ann. Appl. Probab., 22 (2012), pp.~239--284.

\bibitem{MR3055264}
{\sc T.~Ichiba, S.~Pal, and M.~Shkolnikov}, {\em Convergence rates for
  rank-based models with applications to portfolio theory}, Probab. Theory
  Related Fields, 156 (2013), pp.~415--448.

\bibitem{MR2206349}
{\sc H.~Kaise and S.-J. Sheu}, {\em On the structure of solutions of ergodic
  type {B}ellman equation related to risk-sensitive control}, Ann. Probab., 34
  (2006), pp.~284--320.

\bibitem{MR2985170}
{\sc C.~Kardaras and S.~Robertson}, {\em Robust maximization of asymptotic
  growth}, Ann. Appl. Probab., 22 (2012), pp.~1576--1610.

\bibitem{MR2473654}
{\sc S.~Pal and J.~Pitman}, {\em One-dimensional {B}rownian particle systems
  with rank-dependent drifts}, Ann. Appl. Probab., 18 (2008), pp.~2179--2207.

\bibitem{MR1152459}
{\sc Y.~Pinchover}, {\em Large time behavior of the heat kernel and the
  behavior of the {G}reen function near criticality for nonsymmetric elliptic
  operators}, J. Funct. Anal., 104 (1992), pp.~54--70.

\bibitem{MR1326606}
{\sc R.~G. Pinsky}, {\em Positive harmonic functions and diffusion}, vol.~45 of
  Cambridge Studies in Advanced Mathematics, Cambridge University Press,
  Cambridge, 1995.

\bibitem{wong2015universal}
{\sc T.-K.~L. Wong}, {\em Universal portfolios in stochastic portfolio theory},
  arXiv preprint arXiv:1510.02808,  (2015).

\end{thebibliography}

\end{document}